\documentclass{article}

\usepackage[main,final,nonatbib]{neurips_2026}
\usepackage[numbers]{natbib}

\usepackage[T1]{fontenc}
\usepackage{hyperref}
\usepackage{url}
\usepackage{xurl}
\usepackage{microtype}
\usepackage{graphicx}
\usepackage{subcaption}
\usepackage{booktabs}
\usepackage{xcolor}
\usepackage{colortbl}

\definecolor{posgreen}{HTML}{2E8B57}
\definecolor{negred}{HTML}{C0392B}
\definecolor{groupgray}{HTML}{F8F9FA}
\definecolor{algred}{HTML}{C00000}
\definecolor{algblue}{HTML}{0000CD}
\definecolor{harmonyblue}{HTML}{E0F7FA}
\newcommand{\pos}[1]{{\scriptsize\textcolor{posgreen}{+#1}}}
\newcommand{\negc}[1]{{\scriptsize\textcolor{negred}{-#1}}}

\usepackage{multirow}
\usepackage{algorithm}
\usepackage{algorithmic}
\usepackage{float}
\usepackage{placeins}
\usepackage{amsmath}
\usepackage{amssymb}
\usepackage{mathtools}
\usepackage{amsthm}
\usepackage[capitalize,noabbrev]{cleveref}

\theoremstyle{plain}
\newtheorem{theorem}{Theorem}[section]
\newtheorem{proposition}[theorem]{Proposition}
\newtheorem{lemma}[theorem]{Lemma}

\theoremstyle{definition}

\newtheorem{assumption}[theorem]{Assumption}
\theoremstyle{remark}
\newtheorem{remark}[theorem]{Remark}

\title{HARMONY: Hierarchical Anchor Retrieval on Manifold for Oblivious-source Machinery Acoustic Anomaly Detection}

\author{%
  Yuhang Zhang$^{1}$ \quad
  Yang Li$^{2}$ \quad
  Qiang Wu$^{3}$\thanks{Corresponding authors: Qiang Wu
  (\texttt{wuq17@lzu.edu.cn}) and Mingyuan Li
  (\texttt{mingyl@utu.fi}).} \quad
  Mingyuan Li$^{4,5}$\footnotemark[1] \\
  $^{1}$School of Software and Microelectronics, Peking University, Beijing, China \\
  $^{2}$Independent Researcher \\
  $^{3}$Collaborative Innovation Center for Western Ecological Safety, \\
  Lanzhou University, Lanzhou, China \\
  $^{4}$ELLIS Institute of Finland \quad $^{5}$University of Turku \\
  \texttt{zhang\_yuhang@stu.pku.edu.cn} \quad
  \texttt{liyang02@bupt.edu.cn} \\
  \texttt{wuq17@lzu.edu.cn} \quad
  \texttt{mingyl@utu.fi}
}

\begin{document}

\maketitle

\begin{abstract}
%Background
In standard acoustic anomaly detection (ASD), models are usually trained separately for each machine or domain using explicit metadata. 
%Limitation
However, in realistic deployments, machine identifiers are incomplete, unreliable, or simply unavailable. Thus, metadata-free machinery ASD asks a single model to monitor mixed machine populations without machine IDs. 
In this paper, we argue that the main failure mode is not a weak detector family, but representation entanglement. The heterogeneous normal modes overlap in frozen pre-trained feature spaces, making anomaly scoring and approximate retrieval unreliable. We study this problem directly and propose HARMONY (\textbf{H}ierarchical \textbf{A}nchor \textbf{R}etrieval on \textbf{M}anifold for \textbf{O}blivious-source acoustic a\textbf{N}omal\textbf{Y} detection), an anchor-guided geometric repartitioning framework that reorganizes mixed-source features into compact Voronoi regions and reuses this structure for hierarchical retrieval. On DCASE 2020, HARMONY improves mean AUC from 75.40\% to 90.00\% (14.60 percentage points) over the strongest unified baseline and from 83.90\% to 90.00\% over a strong foundation-feature baseline. On a single CPU, its hierarchical retrieval achieves a 7.8$\times$ speedup over exhaustive search with only a 0.20 percentage-point AUC drop. These results suggest that incorporating explicit geometric partitioning can improve both detection accuracy and retrieval efficiency for ASD. Code is available at \url{https://github.com/Zzyh0208/HARMONY}.
\end{abstract}

\section{Introduction}
\label{sec:intro}

In standard acoustic anomaly detection (ASD), models are usually trained separately for each machine or domain using explicit metadata \citep{liu2022anomalous}. This assumption breaks down in realistic deployments, where machine identifiers are incomplete, unreliable, or simply unavailable. The resulting \textit{metadata-free machinery ASD} setting asks a single model to detect anomalies across a mixed stream of machine sounds without knowing which machine generated each sample. This is the practically relevant setting, but it is also where many strong detectors collapse.

A common interpretation is that the solution requires a stronger anomaly detector. We take a different view. When normal data from multiple machine types are pooled together, the dominant problem is often representational: pre-trained feature spaces entangle heterogeneous normal modes, so nearest-neighbor, density-estimation, and regression-based detectors all operate on an ambiguous geometry. Under this mixture, even approximate retrieval structures such as HNSW \citep{malkov2018efficient} only accelerate access to the same entangled neighborhoods. They improve lookup, not the structure being queried.

This leads to the central claim of the paper: in metadata-free machinery ASD, the missing ingredient is not another detector head, but an explicit geometry-aware partition of the mixed normal distribution. If the representation can be reorganized into compact, mode-consistent regions before retrieval, then anomaly scoring becomes more reliable and efficient search becomes possible in the same learned space.

We instantiate this claim with HARMONY, a framework that couples anchor-guided geometric repartitioning with hierarchical retrieval. HARMONY learns a projection that maps mixed-source features into compact Voronoi regions on the hypersphere, then reuses this learned topology as an index that reduces the candidate search space. Rather than treating retrieval as a post-hoc engineering layer, HARMONY makes partitioning and retrieval part of the same modeling decision.

\textbf{Contributions.} Our goal is not to develop a new audio foundation backbone. Instead, we investigate a stricter question: when machine identities are unavailable, can explicit repartitioning of a \textit{frozen} normal feature space restore the separability and efficiency lost in standard unified pipelines? Our main contributions are: 
\textbf{(1) Problem Formulation:} We formulate metadata-free machinery ASD as a representation partitioning problem, providing evidence that manifold entanglement is a bottleneck shared by several detector families. 
\textbf{(2) HARMONY Framework:} We propose an anchor-guided geometric repartitioning framework that reorganizes mixed-source normal data into compact Voronoi regions, intrinsically coupling this structure with hierarchical retrieval. 
\textbf{(3) SoTA Performance and Efficiency:} We provide rigorous analytical and empirical evidence for our approach. On DCASE 2020, HARMONY improves mean AUC by 14.60 percentage points over the strongest unified baseline and by 6.10 percentage points over a foundation-feature baseline, while its hierarchical retrieval achieves a \textbf{7.8}$\times$ CPU speedup over exhaustive search with only a 0.20 percentage-point AUC drop.
\section{Related Work}
\label{sec:related}

\textbf{Feature-Based ASD.} The prevailing unsupervised ASD pipeline is a two-stage representation learning followed by density estimation \citep{alamunsupervised}, with foundation encoders like BYOL-A \citep{niizumi2022byol} as robust feature extractors. Memory-bank methods (e.g., PatchCore) attain SoTA via nearest neighbor search \citep{roth2022towards}, while SimpleNet adds noise injection to train a discriminator on frozen features \citep{liu2023simplenet} and EfficientAD targets scalability through fast feature regression \citep{batzner2024efficientad}. These approaches, however, depend on the geometry of the normal feature distribution \citep{ruff2018deep}; in universal, metadata-free settings, feature distributions from diverse machine types overlap, degrading distance-based discriminability on mixed-source benchmarks.

\textbf{Deep Clustering and Theoretical Foundations.} Unsupervised deep clustering aims to structure latent representations without labels and shares algorithmic foundations with anomaly detection. Methods such as DeepCluster \citep{caron2018deep} and SwAV \citep{caron2020unsupervised} organize latent space by optimizing prototypes. Recent work has formalized the connection to anomaly detection: \citet{fang2024towards} connected representation learning, clustering, and anomaly scoring through a probabilistic mixture model, and \citet{fan2025salvaging} applied class-aware contrastive learning to reduce multi-class confusion. However, a fundamental geometric conflict remains. General representation learning methods like SimCLR \citep{chen2020simple} promote uniformity to maximize information, while anomaly detection requires compactness.

\textbf{Unified Anomaly Detection and Indexing.} The field is rapidly shifting toward unified anomaly detection (UAD), where a single model handles multiple categories or domains. Transformer-based frameworks such as Dinomaly \citep{guo2025dinomaly} adopt a ``less is more'' philosophy to model normal patterns across classes. To cope with missing labels, CADA \citep{guo2024absolute} proposes class-agnostic distribution alignment and RLR \citep{he2024learning} learns unified reference representations to avoid reconstruction shortcuts. For feature modeling, flow-based methods including CFLOW-AD \citep{gudovskiy2022cflow}, UniFlow \citep{zhong2024uniflow} and VQ-Flow \citep{zhou2025vq} apply normalizing flows to capture complex multi-class distributions. Despite their promise, these unified approaches often rely on heavy backbones or complex flow architectures that are computationally costly. Moreover, standard approximate indexing methods such as HNSW \citep{malkov2018efficient} act as passive structures over a frozen feature space and do not resolve the underlying manifold entanglement, producing a trade-off between retrieval precision and inference efficiency.

\section{Preliminaries}
\label{sec:preliminaries}
\textbf{Acoustic Anomaly Detection.} Industrial ASD monitors operating sounds to identify departures from normal machine behavior, such as faults in fans, pumps, sliders, and valves.
\label{sec:task_def}
Let $\mathcal{X} \subset \mathbb{R}^T$ denote the domain of raw audio signals. In the standard unsupervised ASD framework, a system is trained on a dataset $\mathcal{D}_{\text{train}} = \{x_i\}_{i=1}^N$ consisting exclusively of normal samples. The objective is to learn a scoring function $S: \mathcal{X} \to \mathbb{R}$ that quantifies the deviation of a query sample $x_{\text{test}}$ from the learned normal distribution. A decision threshold $t$ is applied such that $x_{\text{test}}$ is classified as anomalous if $S(x_{\text{test}}) > t$. Physically, anomalies in this domain typically manifest as structural deviations, such as broken periodicity, impulsive noises, or spectral energy shifts relative to the normal operational sound.

\textbf{Feature-Based Unsupervised Pipelines.}
\label{sec:feature_asd}
Due to the high dimensionality and temporal redundancy of raw waveforms, the prevailing paradigm operates in a latent feature space. The pipeline typically involves projecting audio spectrograms into an embedding space $\mathcal{Z} \subset \mathbb{R}^D$ via a pre-trained encoder $\phi: \mathcal{X} \to \mathcal{Z}$. Encoders adapted from audio representation learning, such as BYOL-A \citep{niizumi2022byol}, are commonly employed to extract invariant acoustic features. Under this paradigm, the anomaly score $S(x)$ is generally formulated based on density estimation. Common approaches include calculating the Euclidean distance to the nearest normal neighbors stored in a memory bank, or measuring the reconstruction error from a learned probabilistic model. These methods rely on the assumption that normal samples occupy a compact region in $\mathcal{Z}$, while anomalies reside in low-density areas.

\textbf{Metadata-free Machinery ASD.}
\label{sec:universal_setting}
Conventional ASD protocols follow a ``machine-dependent'' approach, training a distinct model $f_m$ for each specific machine ID. While effective in controlled environments, this dependency on explicit metadata creates a scalability bottleneck in industrial scenarios where maintaining thousands of individual models is operationally prohibitive.
This necessitates the metadata-free machinery setting, where a single unified model monitors a diverse fleet of machines without access to machine labels or IDs. Mathematically, the normal training distribution $\mathcal{P}_{\text{norm}}$ is modeled as a mixture of $L$ heterogeneous acoustic modes:
\begin{equation}
    \mathcal{P}_{\text{norm}}(x) = \sum_{k=1}^L \pi_k \mathcal{P}_k(x), \quad \text{s.t.} \quad \sum_{k=1}^L \pi_k = 1,
    \label{eq:mixture_model}
\end{equation}
where $\mathcal{P}_k$ represents the distribution of the $k$-th latent mode (e.g., fan, pump, slider) and $\pi_k$ is its prior probability. In this setting, the specific mode index $k$ for a given sample is unknown during training, requiring the model to capture the global structure of the mixture solely from unlabeled data. The latent mode count $L$ is unknown and need not equal the anchor count $K$ used by HARMONY. Machine metadata are used only to group test scores for evaluation.
\section{Method}
\label{sec:method}

For raw audio $x$, we write $g_\theta(x)$ as shorthand for $g_\theta(\phi(x))$; on encoded inputs we use $g_\theta(z)$ explicitly. We propose HARMONY to resolve feature space ambiguities in metadata-free machinery ASD through active geometric reconstruction. As shown in \cref{fig:framework} and \cref{alg:training}, the framework couples intrinsic subspace estimation, anchor-guided disentanglement, and hierarchical retrieval in a single pipeline.

% Figure 1 Placeholder
\begin{figure*}[t]
    \centering
    \includegraphics[width=1.0\linewidth]{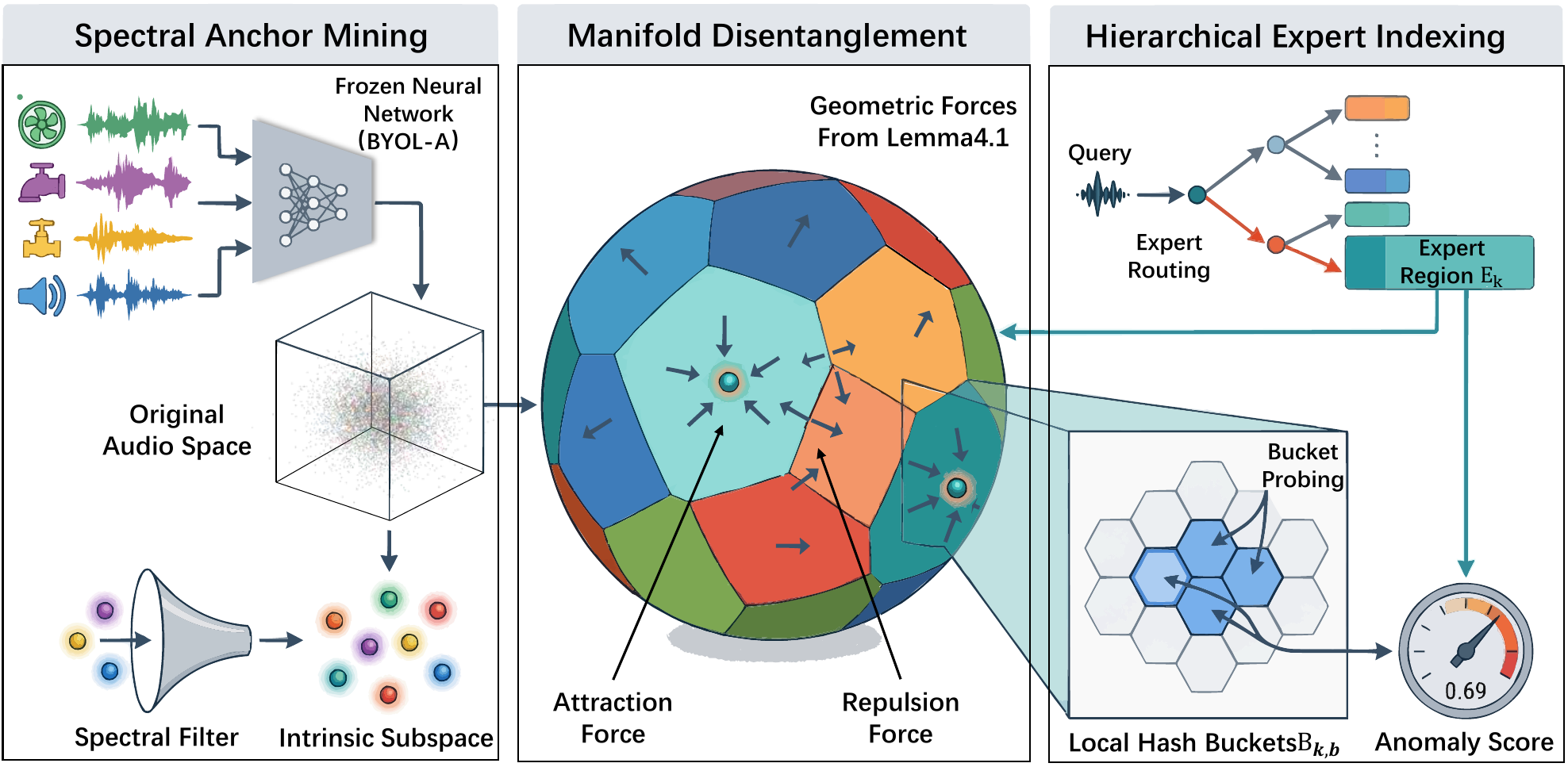}
    \caption{\textbf{Overview of HARMONY.} Spectral anchors define the coarse topology, anchor-guided projection compacts mixed-source features, and the same geometry is reused for hierarchical retrieval.}
    \label{fig:framework}
\end{figure*}

\subsection{Geometric Problem Analysis}
\label{sec:problem_formulation}

Let $\mathcal{Z} = \phi(\mathcal{X})$ be the frozen feature space derived from the pre-trained encoder. The key obstacle is \emph{manifold entanglement}: supports of distinct heterogeneous normal modes $\{\mathcal{P}_k\}_{k=1}^L$ overlap in $\mathcal{Z}$,
\begin{equation}
    \text{Supp}(\phi_{\#}\mathcal{P}_i) \cap \text{Supp}(\phi_{\#}\mathcal{P}_j) \neq \emptyset, \quad \text{for some } i \neq j.
\end{equation}
where $\phi_{\#}\mathcal{P}_i$ denotes the distribution of encoded features $\phi(x)$ when $x\sim\mathcal{P}_i$ (the pushforward measure). This entanglement breaks the compactness assumption used by standard distance-based detectors. Our objective is therefore to learn a projection $g_\theta: \mathcal{Z} \to \mathcal{H} \subset \mathbb{S}^{d-1}$ that reconstructs the mixed distribution into locally compact, globally separated Voronoi regions.

\subsection{Spectral Anchor Mining}

An anchor $a_k$ is a unit-norm prototype that defines a local acoustic region; it is not a machine label, and one machine can occupy several regions. Random anchors are unreliable in a high-dimensional mixed feature space, so we first estimate a lower-dimensional tangent space that captures the dominant acoustic modes while suppressing orthogonal noise. Given covariance matrix $\mathbf{\Sigma} = \frac{1}{N-1} \sum_{i=1}^N (z_i - \bar{z})(z_i - \bar{z})^\top$, we select the smallest intrinsic dimension $d'$ satisfying
\begin{equation}
    d' = \min \left\{ k : \frac{\sum_{j=1}^k \lambda_j}{\sum_{j=1}^D \lambda_j} \ge \eta \right\}.
\end{equation}
where $\lambda_j$ is the $j$-th largest covariance eigenvalue, $\eta$ is the retained-variance threshold, and $D$ is the ambient dimension of the full input space $\mathcal{Z}$. Projecting onto the principal eigenvectors gives filtered features $\tilde{z}_i = U_{d'}^\top (z_i - \bar{z})$. We then initialize $K$ anchors with a diversity-preserving process $\mathcal{G}$ (e.g., K-Means++) on $\{\tilde{z}_i\}$:
\begin{equation}
    \mathcal{A}' \leftarrow \mathcal{G}(\{\tilde{z}_i\}_{i=1}^N, K), \quad \text{with } \mathcal{A}' \subset \mathbb{R}^{d'}.
\end{equation}
We map these prototypes back to $\mathbb{R}^d$ by padding or truncation to obtain the final anchor set $\mathcal{A}$. This provides a stable topological skeleton for the subsequent geometric repartitioning.

\begin{algorithm}
\caption{HARMONY Training and Construction Procedure}
\label{alg:training}
\small % 统一使用 \small 字体，与表格一致
\begin{algorithmic}[1] % [1] 表示开启行号显示
\REQUIRE Unlabeled acoustic data $\mathcal{D}=\{x_i\}_{i=1}^N$, Frozen encoder $\phi$, Anchor count $K$
\ENSURE Trainable projection $g_\theta$, Hierarchical index $\mathcal{T}=\{\mathcal{A}, \mathcal{B}\}$
\STATE ~ % 空行留白
\STATE \textbf{// Phase 1: Subspace Estimation (Setup)}
\STATE \textbf{Extract} raw features $\mathcal{Z} \leftarrow \phi(\mathcal{D})$
\STATE \textbf{Compute} covariance matrix $\mathbf{\Sigma}$ and top eigenvectors $U_{d'}$
\STATE \textbf{Project} $\tilde{\mathcal{Z}} \leftarrow \mathcal{Z} U_{d'}$ \hfill \textit{[Orthogonal noise filtering]}
\STATE \textbf{Initialize} topological prototypes $\mathcal{A}' \leftarrow \text{Spectral}(\tilde{\mathcal{Z}}, K)$
\STATE \textbf{Align} dimensions: Map $\mathcal{A}' \to \mathcal{A} \in \mathbb{R}^d$ via zero-padding
\STATE \textbf{Normalize} anchors to hypersphere: $a_k \leftarrow a_k / \|a_k\|_2$
\STATE ~
\STATE \textcolor{algred}{\textbf{// Phase 2: Active Manifold Disentanglement}}
\FOR{$epoch = 1$ \textbf{to} $E$}
    \FOR{each mini-batch $z_b \in \mathcal{Z}$}
        \STATE \textbf{Compute} soft assignment $p(k|z_b)$ and target pseudo-label $y^*_b$
        \STATE \textcolor{algred}{$\mathcal{L}_{align} \leftarrow -\frac{1}{|B|}\sum \log p(y^*_b|z_b)$ \hfill \textit{[Core push-pull injection]}}
        \STATE \textbf{Update} projection $g_\theta$ via Adam to minimize $\mathcal{L}_{align}$
    \ENDFOR
\ENDFOR
\STATE ~
\STATE \textcolor{algblue}{\textbf{// Phase 3: Hierarchical Indexing (Recovery)}}
\STATE \textbf{Partition} hypersphere $\mathcal{H}$ into expert Voronoi regions $\{E_k\}_{k=1}^K$ using $\mathcal{A}$
\STATE \textcolor{algblue}{\textbf{Build} local hash buckets $\mathcal{B}_{k}$ within each $E_k$ \hfill \textit{[Fast retrieval setup]}}
\STATE ~
\STATE \textbf{return} $g_\theta, \mathcal{T}=\{\mathcal{A}, \mathcal{B}\}$
\end{algorithmic}
\end{algorithm}

\subsection{Anchor-Guided Manifold Disentanglement}
\label{subsec:disentanglement}

Given the extracted topological skeleton $\mathcal{A}=\{a_k\}_{k=1}^K$, we restructure the latent representation space by aligning samples to anchor-defined acoustic modes. This is achieved by learning an anchor-aligned partition of the hyperspherical feature space via a parametric projection $g_\theta: \mathcal{Z} \rightarrow \mathbb{S}^{d-1}$.

\paragraph{Surrogate Alignment Objective.}
For each sample $z_i$, we quantify its affinity to anchors using a temperature-scaled Boltzmann distribution:
\begin{equation}
    p(k|z_i; \theta) = 
    \frac{\exp(\langle g_\theta(z_i), a_k \rangle / \tau)}
    {\sum_{j=1}^K \exp(\langle g_\theta(z_i), a_j \rangle / \tau)},
    \label{eq:prob_final}
\end{equation}
where $\tau>0$ controls assignment sharpness. For fixed features and anchors, changing a shared positive temperature preserves the maximum-similarity assignment and hence the hard Voronoi boundary. During learning, it changes assignment probabilities and the gradient scale, affecting compactness under unequal mode densities. Each sample is pseudo-labeled by its most probable anchor, $y_i^* = \arg\max_k p(k|z_i; \theta)$. We optimize the alignment loss
\begin{equation}
    \mathcal{L}_{align}(\theta)
    = \frac{1}{N}\sum_{i=1}^N \ell_i(z_i; \theta), 
    \quad 
    \ell_i = -\log p(y_i^*|z_i; \theta),
    \label{eq:loss_final}
\end{equation}
which serves as a smooth surrogate objective encouraging anchor-consistent hyperspherical clustering.

\paragraph{Gradient-Induced Geometric Forces.}
Let $u_i = g_\theta(z_i)$ denote the projected feature. For a sample assigned to anchor $k=y_i^*$, the ambient gradient of $\ell_i$ with respect to $u_i$, holding the discrete pseudo-label fixed away from assignment ties, is:
\begin{equation}
    \nabla_{u_i} \ell_i
    = -\frac{1}{\tau} a_k
    + \frac{1}{\tau} \sum_{j=1}^K p(j|z_i)\, a_j.
    \label{eq:grad_raw_final}
\end{equation}

\begin{lemma}[Gradient Decomposition]
\label{lem:gradient}
The negative gradient direction decomposes into two competing geometric forces:
\begin{equation}
    -\nabla_{u_i} \ell_i
    = \underbrace{\frac{1}{\tau}(1-p_i^*)a_{y_i^*}}_{\text{Attraction}}
    - \underbrace{\frac{1}{\tau}\sum_{j\neq y_i^*} p_{ij} a_j}_{\text{Repulsion}},
    \label{eq:gradient_final}
\end{equation}
where $p_i^*$ denotes $p(y_i^*|z_i)$ and $p_{ij}$ denotes $p(j|z_i)$.
\end{lemma}

The attraction term pulls samples toward assigned anchors based on prediction uncertainty, while the repulsion term pushes them away from competitors. The same anchors therefore couple representation learning, local partitioning, and inference routing. The detailed formula derivation is in the Appendix~\ref{deriva_lemma4.1}.

\paragraph{Optimization and Separation.}
The alignment objective encourages compact assigned regions. Smoothness gives a descent property, while a uniform within-anchor residual is an additional condition for the separation result. The following result states the resulting separation conditional on bounded within-anchor residuals; Appendix \ref{app:theory} gives its proof and optimization assumptions.

\begin{theorem}[Loss-Induced Manifold Separation]
\label{thm:separation}
Let $\Delta_{\min} = \min_{k\neq k'} \|a_k-a_{k'}\|_2$ denote the minimum Euclidean separation between anchors. Suppose normal samples assigned to anchor $a_k$ remain within residual $\delta_{train}$ of that anchor after training, and let $\eta_{gen}$ bound the additional test-time deviation. Then any two samples $x,x'$ assigned to distinct anchors satisfy:
\begin{equation}
    \|g_\theta(x) - g_\theta(x')\|_2 
    \ge 
    \Delta_{\min} - 2(\delta_{train} + \eta_{gen}).
\end{equation}
\end{theorem}

Theorem~\ref{thm:separation} formalizes the intended geometric effect of the alignment objective. In particular, once training produces small within-anchor residuals, samples assigned to different anchors remain separated by a positive margin in the learned space. Appendix \ref{app:theory} distinguishes the optimization descent property from the residual assumptions needed for this bound.

Crucially, this geometric structure naturally facilitates anomaly scoring. Since the projection network $g_\theta$ is optimized exclusively on normal data to align with specific anchors, normal samples are compressed into tight $\varepsilon$-balls around these anchors. Anomalies are scored by their distance to the projected normal memory bank. Their separation from normal samples is an empirical property, rather than a consequence guaranteed by the normal-only alignment objective. We evaluate this score separation empirically, together with anchor occupancy stability, in Appendix \ref{app:empirical_theory}.

\subsection{Hierarchical Expert Indexing}

To address the scalability challenge, we transition from linear scanning to a structured hierarchical retrieval. Our index structure $\mathcal{T}$ is intrinsically derived from the learned Voronoi topology, ensuring that the retrieval path aligns with the feature geometry.

% We discretize the latent manifold $\mathcal{H}$ into a two-tier structure:
% \begin{itemize}
%     \item \textbf{Level 1 (Expert Regions):} The manifold is partitioned into $K$ disjoint regions $\{E_k\}_{k=1}^K$, where $E_k = \{z \in \mathcal{H} : \arg\max_j \langle z, a_j \rangle = k\}$.
%     \item \textbf{Level 2 (Local Buckets):} Within each $E_k$, we further partition samples into $B$ sub-regions $\{\mathcal{B}_{k,b}\}_{b=1}^B$ using local vector quantization.
% \end{itemize}

% During inference, given a query $q = g_\theta(x_q)$, we employ a coarse-to-fine routing strategy:
% \begin{enumerate}
%     \item \textbf{Expert Routing:} Identify the set of candidate experts $\mathcal{K} \subset \{1, \dots, K\}$ such that:
%     \begin{equation}
%         \mathcal{K} = \text{Top-}K_1 \left( \{ \langle q, a_k \rangle \}_{k=1}^K \right).
%     \end{equation}
%     \item \textbf{Bucket Probing:} Within the selected experts, retrieve the candidate buckets $\mathcal{S}$ using centroids $c_{k,b}$:
%     \begin{equation}
%         \mathcal{S} = \bigcup_{k \in \mathcal{K}} \text{Top-}K_2 \left( \{ \langle q, c_{k,b} \rangle \}_{b=1}^B \right).
%     \end{equation}
%     \item \textbf{Approximate Scoring:} Compute the practical anomaly score $\hat{s}(q)$ by considering only samples $z \in \mathcal{S}$:
%     \begin{equation}
%         \hat{s}(q) = \min_{z \in \mathcal{S}} \|q - z\|_2.
%     \end{equation}
% \end{enumerate}
We discretize the latent manifold $\mathcal{H}$ into a hierarchical two-tier structure for efficient retrieval. At the coarse level, the space is partitioned into $K$ disjoint expert regions $\{E_k\}_{k=1}^K$, where each region is defined by maximum similarity to a set of anchors $\{a_k\}$, i.e., $E_k = \{z \in \mathcal{H} : \arg\max_j \langle z, a_j \rangle = k\}$. At the fine level, each expert region is further subdivided into $B$ local buckets $\{\mathcal{B}_{k,b}\}_{b=1}^B$ via vector quantization, enabling finer-grained organization within each region.

At inference time, given a query representation $q = g_\theta(x_q)$, we adopt a coarse-to-fine routing strategy to efficiently approximate nearest-neighbor search. First, we perform \emph{expert routing} by selecting the top-$K_1$ most relevant regions based on similarity scores $\langle q, a_k \rangle$. Next, within the selected experts, we conduct \emph{bucket probing} by retrieving the top-$K_2$ closest buckets according to their centroids $\{c_{k,b}\}$. Finally, we compute an approximate anomaly score by restricting the search to the candidate set $\mathcal{S}$ formed by these buckets, yielding $\hat{s}(q) = \min_{z \in \mathcal{S}} \|q - z\|_2$. This hierarchical design significantly reduces the search space while preserving retrieval accuracy.
This pruning strategy reduces the search complexity from $\mathcal{O}(N \cdot d)$ to $\mathcal{O}((K + K_1 B + \frac{K_1 K_2 N}{K B}) \cdot d)$, achieving a significant constant-factor speedup, though remaining asymptotically linear $\mathcal{O}(N)$ relative to the database size $N$. Crucially, the following quantization bound characterizes score distortion under candidate coverage. To quantify this, let $\mathcal{H}_{norm} = \{g_\theta(z) \mid z \in \mathcal{Z}\}$ denote the projected normal database. We define the ideal exhaustive-search score as the global minimum distance in this latent space: $s(q) = \min_{h \in \mathcal{H}_{norm}} \|q - h\|_2$.

\begin{theorem}[Bounded Retrieval Error]
\label{thm:robustness}
Given the ideal exhaustive-search score $s(q)$ and the practical score $\hat{s}(q)$ obtained via hierarchical retrieval, suppose the candidate set contains a point in the same radius-$\epsilon$ bucket as a true nearest neighbor. Then the retrieval error is strictly bounded:
\begin{equation}
    |\hat{s}(q) - s(q)| \le 2\epsilon.
\end{equation}
\end{theorem}

The bound is conditional on candidate coverage; expert selection alone does not guarantee bucket coverage. It shows that hierarchical routing is a controlled approximation to exhaustive search in the learned space: the theorem bounds pointwise score distortion, while the resulting detection accuracy under this approximation is evaluated empirically in \cref{sec:experiments}.
\section{Experiments}
\label{sec:experiments}

\makeatletter
\@ifundefined{color@takeawaybg}{\definecolor{takeawaybg}{RGB}{232,241,246}}{}
\@ifundefined{color@groupgray}  {\definecolor{groupgray}{RGB}{246,246,246}}{}
\@ifundefined{color@harmonyblue}{\definecolor{harmonyblue}{RGB}{222,244,248}}{}
\@ifundefined{color@bestgreen}{\definecolor{bestgreen}{RGB}{0,130,55}}{}
\@ifundefined{color@secondblue}{\definecolor{secondblue}{RGB}{36,91,185}}{}
\makeatother
\newlength{\takeawaywidth}
\setlength{\takeawaywidth}{0.97\columnwidth}
\newcommand{\takeawaybox}[2]{%
    \par\smallskip
    \begin{center}
    {%
        \setlength{\fboxsep}{5pt}%
        \colorbox{takeawaybg}{%
            \begin{minipage}{\dimexpr\takeawaywidth-2\fboxsep\relax}
            \small \textbf{Takeaway (#1).} #2
            \end{minipage}
        }%
    }%
    \end{center}
    \vspace{-1mm}
}
\newcommand{\best}[1]{\textbf{\textcolor{bestgreen}{#1}}}
\newcommand{\second}[1]{\underline{\textcolor{secondblue}{#1}}}

\subsection{Experimental Setup}
We evaluate HARMONY primarily on DCASE 2020 Task 2 \citep{koizumi2020description} under a metadata-free machinery protocol: normal samples from all 23 machine IDs are pooled without metadata during training. We also evaluate MIMII \citep{purohit2019mimii} and DCASE 2022 \citep{dcase2022task2}, including noise and source-to-target shifts. We utilize a frozen BYOL-A encoder \citep{niizumi2022byol} ($D=3072$ dim) on 16kHz log-Mel spectrograms. During training, HARMONY employs $K=8$ anchors and $B=64$ local buckets, optimized via Adam ($10^{-3}$ lr) for 20 epochs. For the hierarchical retrieval during inference, we set the routing hyperparameters to $K_1=3$ and $K_2=10$. Anchor count and temperature are selected using the Silhouette Score on pooled normal data, without anomaly labels or test AUC. Detailed baseline adaptations are provided in Appendix \ref{app:implementation}.

The empirical study addresses four core questions: \textbf{RQ1} confirms mixed-source entanglement as the primary bottleneck in metadata-free machinery ASD; \textbf{RQ2} compares HARMONY's active geometric disentanglement against unified baselines; \textbf{RQ3} ablates the proposed mechanism and evaluates stability under perturbations; and \textbf{RQ4} assesses whether the hierarchical index achieves efficient CPU retrieval without accuracy degradation compared to exact search.

\subsection{RQ1: Is Representation Entanglement the Main Obstacle?}
\label{sec:rq1}

We first ask whether the universal setting fails because existing detectors are too weak, or because the mixed-source geometry itself is problematic. To test this, we use HARMONY only as a frozen router: it partitions the unlabeled universal dataset into semantically consistent regions, after which standard anomaly detectors are trained per region. If these same detectors recover strongly after repartitioning, then the bottleneck lies in the mixed representation rather than in the detector family.

\begin{table}[ht]
\normalsize
\renewcommand{\arraystretch}{1.15}
\setlength{\tabcolsep}{4pt}
\caption{Unified vs. per-region training (\%). Baseline AUC gains ($\Delta$) under per-region training expose mixed-source geometry as the core bottleneck, yet partitioning severely inflates costs. HARMONY natively achieves the optimal accuracy-efficiency balance.}
\label{tab:disentanglement_verification}
\centering
\begin{tabular}{l cccc cc}
\toprule
\multirow{2}{*}{Method} & \multicolumn{2}{c}{Unified Setting} & \multicolumn{2}{c}{Per-Region Setting} & \multicolumn{2}{c}{Cost Transition (Uni $\to$ Reg)} \\
\cmidrule(lr){2-3} \cmidrule(lr){4-5} \cmidrule(lr){6-7}
& AUC & pAUC & AUC ($\Delta$) & pAUC & Train Time (h) & Latency (ms) \\
\midrule
PatchCore   & 62.10 & 56.20 & 91.80\pos{29.70} & 78.90 & 0 $\to$ 0 & 41.8 $\to$ \textbf{30.6} \\
Audio-PC    & 83.90 & 71.30 & \textbf{93.10}\pos{9.20} & \textbf{80.40} & 0 $\to$ 0 & 67.5 $\to$ 52.1 \\
EfficientAD & 66.30 & 58.70 & 88.90\pos{22.60} & 76.10 & 2.4 $\to$ \textcolor{negred}{18.5} & 6.3 $\to$ 6.3 \\
Audio-MAE   & 86.10 & 72.80 & 92.40\pos{6.30} & 79.60 & 14.2 $\to$ \textcolor{negred}{108.0} & 82.4 $\to$ 82.4 \\
MCDD        & 78.30 & 66.90 & 89.30\pos{11.00} & 76.80 & 7.1 $\to$ \textcolor{negred}{55.6} & 54.7 $\to$ 44.9 \\
\midrule
\rowcolor{harmonyblue}
HARMONY (IVF) & \textbf{89.80} & \textbf{74.70} & \multicolumn{2}{c}{\textit{Native Disentanglement}} & \textbf{0} (Unified) & \textbf{5.4} (Unified) \\
\bottomrule
\end{tabular}
\end{table}

\cref{tab:disentanglement_verification} shows that every strong baseline improves once the mixed dataset is repartitioned into more homogeneous regions. Audio-MAE gains 6.30 AUC points, Audio-PC gains 9.20 AUC points, and even simple PatchCore gains a substantial 29.70 AUC points. This is problem-level evidence that the dominant bottleneck is entangled mixed-source geometry rather than detector family alone. Furthermore, t-SNE visualizations in Appendix \ref{app:tsne_visualization} visually confirm how HARMONY successfully separates these initially entangled features into distinct Voronoi regions.

\subsection{RQ2: Does Active Repartitioning Improve Metadata-free machinery ASD Accuracy?}
\label{sec:rq2}

Having established that the problem is geometric, we next test whether HARMONY's \emph{active} repartitioning improves accuracy over strong unified baselines rather than merely reusing a stronger detector or post-hoc index.

\begin{table*}[t]
\normalsize
\renewcommand{\arraystretch}{1.15}
\setlength{\tabcolsep}{2.5pt}
\caption{Metadata-free machinery anomaly detection accuracy (AUC, \%) on DCASE 2020. Absolute sub-metric changes (\textcolor{posgreen}{+}/\textcolor{negred}{-}) and relative average improvements (\%) are computed against the PatchCore baseline. All methods share the same pre-trained features.}
\label{tab:dcase_results}
\begin{tabular}{l ccccccc}
\toprule
\multirow{2}{*}{Methods} & \multicolumn{6}{c}{Machine Types} & \multirow{2}{*}{Average} \\
\cmidrule(lr){2-7}
& ToyCar & ToyConv & Fan & Pump & Slider & Valve & \\
\midrule
PatchCore & 62.80 & 60.60 & 62.20 & 57.20 & 67.80 & 61.50 & 62.10 \\
\rowcolor{groupgray}
\multicolumn{8}{c}{\textit{Classic Adaptation}} \\
SimpleNet   & 63.20\pos{0.40} & 62.10\pos{1.50} & 56.70\negc{5.50} & 58.00\pos{0.80} & 66.30\negc{1.50} & 62.20\pos{0.70} & 61.40{\scriptsize\textcolor{negred}{(-1.1\%)}} \\
EfficientAD & 65.80\pos{3.00} & 65.60\pos{5.00} & 65.70\pos{3.50} & 68.50\pos{11.30} & 66.50\negc{1.30} & 65.60\pos{4.10} & 66.30{\scriptsize\textcolor{posgreen}{(+6.8\%)}} \\
\rowcolor{groupgray}
\multicolumn{8}{c}{\textit{Foundation Embeddings}} \\
Gen-Rep     & 82.30\pos{19.50} & 81.10\pos{20.50} & 79.40\pos{17.20} & 79.80\pos{22.60} & 81.20\pos{13.40} & 79.90\pos{18.40} & 80.60{\scriptsize\textcolor{posgreen}{(+29.8\%)}} \\
Audio-PC    & 85.40\pos{22.60} & 84.10\pos{23.50} & 82.50\pos{20.30} & 83.30\pos{26.10} & 85.10\pos{17.30} & 83.20\pos{21.70} & 83.90{\scriptsize\textcolor{posgreen}{(+35.1\%)}} \\
\rowcolor{groupgray}
\multicolumn{8}{c}{\textit{Unified Models}} \\
CADA        & 73.30\pos{10.50} & 73.20\pos{12.60} & 78.30\pos{16.10} & 72.60\pos{15.40} & 73.10\pos{5.30} & 75.90\pos{14.40} & 74.50{\scriptsize\textcolor{posgreen}{(+20.0\%)}} \\
UniFlow     & 75.20\pos{12.40} & 78.70\pos{18.10} & 75.40\pos{13.20} & 72.80\pos{15.60} & 75.30\pos{7.50} & 75.80\pos{14.30} & 75.40{\scriptsize\textcolor{posgreen}{(+21.4\%)}} \\
Dinomaly    & 74.90\pos{12.10} & 73.80\pos{13.20} & 74.70\pos{12.50} & 76.00\pos{18.80} & 74.90\pos{7.10} & 76.80\pos{15.30} & 75.30{\scriptsize\textcolor{posgreen}{(+21.3\%)}} \\
\rowcolor{groupgray}
\multicolumn{8}{c}{\textit{Disentanglement Baselines}} \\
GRHD        & 66.20\pos{3.40} & 63.90\pos{3.30} & 65.10\pos{2.90} & 63.40\pos{6.20} & 65.80\negc{2.00} & 65.30\pos{3.80} & 65.00{\scriptsize\textcolor{posgreen}{(+4.7\%)}} \\
MCDD        & 79.10\pos{16.30} & 78.30\pos{17.70} & 76.90\pos{14.70} & 77.80\pos{20.60} & 79.40\pos{11.60} & 78.20\pos{16.70} & 78.30{\scriptsize\textcolor{posgreen}{(+26.1\%)}} \\
AudDSR      & 81.20\pos{18.40} & 79.80\pos{19.20} & 78.50\pos{16.30} & 79.10\pos{21.90} & 81.30\pos{13.50} & 80.20\pos{18.70} & 80.00{\scriptsize\textcolor{posgreen}{(+28.8\%)}} \\
\rowcolor{groupgray}
\multicolumn{8}{c}{\textit{Proposed Method}} \\
\rowcolor{harmonyblue}
HARMONY     & \textbf{98.00}\pos{35.20} & \textbf{89.70}\pos{29.10} & \textbf{88.70}\pos{26.50} & \textbf{89.70}\pos{32.50} & \textbf{87.60}\pos{19.80} & \textbf{90.00}\pos{28.50} & \textbf{90.00}{\scriptsize\textcolor{posgreen}{(+44.9\%)}} \\
\bottomrule
\end{tabular}
\end{table*}

\textbf{Baselines.} We compare against four baseline families: (1) classic adaptation methods PatchCore \citep{roth2022towards}, SimpleNet \citep{liu2023simplenet}, and EfficientAD \citep{batzner2024efficientad}; (2) foundation-feature pipelines Gen-Rep \citep{saengthong2025deep} and Audio-PC; (3) unified models CADA \citep{guo2024absolute}, UniFlow \citep{zhong2024uniflow}, and Dinomaly \citep{guo2025dinomaly}; and (4) disentanglement-oriented baselines GRHD \citep{guan2025disentangling}, MCDD \citep{lee2020multi}, and AudDSR \citep{zavrtanik2024anomalous}.

\cref{tab:dcase_results} shows HARMONY reaches 90.00\% mean AUC, improving by 14.60 percentage points over the strongest unified baseline UniFlow (75.40\%) and by 6.10 percentage points over the foundation-feature baseline Audio-PC (83.90\%). On the average AUC column, this corresponds to a +7.3\% relative gain over the second-best method. Direct controls in Appendix \ref{app:direct_control} separate representation learning from indexing: vMF soft prototypes reach 85.71\% routed AUC and HARMONY reaches 89.80\%; applying HNSW to the HARMONY representation gives 89.71\%. Thus, the main accuracy gain comes from the representation, while the shared index improves routing recall. Detailed multi-seed stability results, including variance (mean $\pm$ std over 5 seeds) and pAUC metrics, are provided in Appendix \ref{app:multi_seed}.

\paragraph{Cross-dataset comparison.}
We compare the strongest feature baseline and two unified models on MIMII and DCASE 2022 in \cref{tab:cross_dataset}. The latter introduces section-level source-to-target shifts. HARMONY improves over Audio-PC by 4.84 AUC points on MIMII and 4.52 points overall on DCASE 2022, with a 4.27-point gain on the target domain. Full MIMII baseline comparisons and worst-group metrics appear in Appendix \ref{app:cross_dataset}.
\begin{table}[htbp]
\centering
\small
\setlength{\tabcolsep}{4pt}
\caption{Cross-dataset AUC (\%, mean $\pm$ standard deviation over five seeds). MIMII uses the $-6$ dB setting. DCASE 2022 is evaluated with machine metadata withheld from the model.}
\label{tab:cross_dataset}
\begin{tabular}{lccc}
\toprule
Method & MIMII & DCASE 2022 & Target \\
\midrule
Audio-PC & $79.73\pm0.44$ & $77.84\pm0.46$ & $73.72\pm0.59$ \\
UniFlow & $71.24\pm0.79$ & $72.69\pm0.81$ & $68.43\pm0.98$ \\
Dinomaly & $72.06\pm0.71$ & $73.21\pm0.73$ & $68.97\pm0.91$ \\
HARMONY & $84.57\pm0.58$ & $82.36\pm0.52$ & $77.99\pm0.68$ \\
\bottomrule
\end{tabular}
\end{table}

\subsection{RQ3: Are the Gains Attributable and Robust?}
\label{sec:rq3}

We next test whether the gains come from the proposed mechanism itself and whether they remain stable when the setting changes.

\begin{table}[ht]
\normalsize
\renewcommand{\arraystretch}{1.15}
\setlength{\tabcolsep}{4pt}
\centering
\caption{Component ablation study on DCASE 2020 (\%). The full flat-search model serves as the baseline. Performance drops ($\Delta$) are highlighted in \textcolor{negred}{red}. The final deployed HARMONY (Hierarchical) maintains near-lossless accuracy.}
\label{tab:ablation}
\begin{tabular}{l ccccccc c}
\toprule
\multirow{2}{*}{Variant} & \multicolumn{6}{c}{Machine Types} & \multirow{2}{*}{Average} & \multirow{2}{*}{$\Delta$} \\
\cmidrule(lr){2-7}
& ToyCar & ToyConv & Fan & Pump & Slider & Valve & & \\
\midrule
\rowcolor{harmonyblue}
Full Model     & \textbf{98.00} & \textbf{89.70} & \textbf{88.70} & \textbf{89.70} & \textbf{87.60} & \textbf{90.00} & \textbf{90.00} & - \\
\rowcolor{groupgray}
\multicolumn{9}{c}{\textit{Ablated Components}} \\
w/o Soft Assignment & 96.50 & 87.00 & 85.00 & 87.50 & 85.20 & 84.40 & 87.60 & \textcolor{negred}{-2.40} \\
w/o Spectral Init    & 89.50 & 76.50 & 78.20 & 82.50 & 79.00 & 83.10 & 81.40 & \textcolor{negred}{-8.60} \\
w/o Projection      & 64.10 & 56.50 & 58.80 & 69.20 & 62.30 & 63.50 & 62.40 & \textcolor{negred}{-27.60} \\
\midrule
\rowcolor{harmonyblue}
HARMONY   & 97.80 & 89.50 & 88.50 & 89.50 & 87.40 & 89.80 & 89.80 & \textcolor{negred}{-0.20} \\
\bottomrule
\end{tabular}
\end{table}

\cref{tab:ablation} attributes the performance gain to the proposed geometry. Removing the projection collapses the mean AUC from 90.00\% to 62.40\%, while replacing spectral initialization or soft assignment costs 8.60 and 2.40 AUC points respectively. The 0.20-point gap between flat and hierarchical retrieval further suggests that the indexing approximation is not the source of the accuracy gain.

\begin{figure}[t]
    \centering
    \begin{subfigure}[t]{0.48\columnwidth}
        \centering
        \includegraphics[width=\linewidth]{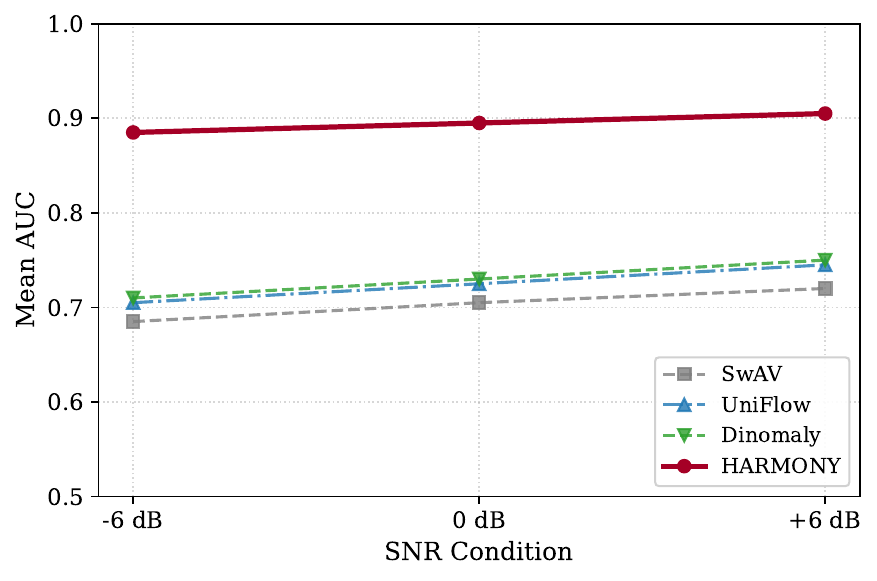}
        \caption{Robustness under dynamic SNR.}
        \label{fig:snr_robustness}
    \end{subfigure}
    \hfill
    \begin{subfigure}[t]{0.48\columnwidth}
        \centering
        \includegraphics[width=\linewidth]{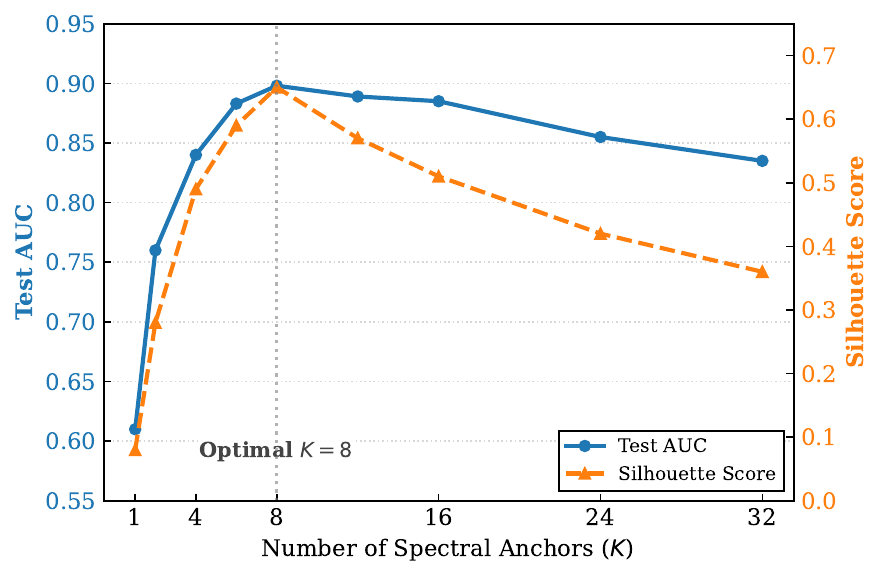}
        \caption{Unsupervised model selection.}
        \label{fig:k_selection}
    \end{subfigure}
    \caption{Robustness and stability. Left: HARMONY remains stable as SNR drops to $-6$ dB, unlike degrading baselines. Right: Silhouette Score tracks AUC, enabling unsupervised model selection.}
    \label{fig:robustness_and_selection}
\end{figure}

\cref{fig:robustness_and_selection} addresses two further concerns. Under dynamic MIMII noise conditions, HARMONY shows a flatter performance curve than baseline methods, and the Silhouette Score tracks downstream AUC well enough to support metadata-free model selection. Additional backbone-transfer results, stress tests, and sensitivity analyses appear in Appendices \ref{app:verification}, \ref{app:stress_tests}, and \ref{app:hyperparameters}.

\paragraph{Geometry of the learned regions.}
\cref{tab:qe_margin} measures average cosine-distance proxies for quantization error and assignment margin. Projection lowers the former from 0.842 to 0.181 and raises the latter from 0.418 to 0.852. The accompanying Silhouette increase supports the interpretation that the accuracy gain coincides with more compact, better-separated normal regions. These average proxies describe observed geometry; they do not substitute for the uniform Euclidean residual assumptions in Theorem~\ref{thm:separation}.

\begin{table}[ht]
\normalsize
\renewcommand{\arraystretch}{1.15}
\setlength{\tabcolsep}{8pt}
\caption{Empirical validation of geometric metrics (measured via aligned cosine distance).}
\label{tab:qe_margin}
\centering
\begin{tabular}{l ccccc}
\toprule
\textbf{Feature Space} & \textbf{QE ($\downarrow$)} & \textbf{Margin ($\uparrow$)} & \textbf{Ratio ($\uparrow$)} & \textbf{Silh. Score ($\uparrow$)} & \textbf{AUC ($\uparrow$, \%)} \\
\midrule
Raw BYOL-A & 0.842 & 0.418 & 0.496 & 0.331 & 62.40 \\
\rowcolor{harmonyblue}
HARMONY   & \textbf{0.181} & \textbf{0.852} & \textbf{4.707} & \textbf{0.824} & \textbf{90.00} \\
\bottomrule
\end{tabular}
\end{table}

\paragraph{Selection across mixture complexity.}
\cref{tab:selection_main} summarizes normal-only selection as two to six machine types are pooled. The selected count changes from 4 to 8 and remains within 0.19--0.37 AUC points of the retrospective oracle. In the separate long-tailed mixtures, the selected configuration changes from $K=8,\tau=0.10$ on balanced data to $K=16,\tau=0.05$ at imbalance ratio 20, with 84.83\% macro AUC and 79.36\% tail-type AUC. Appendix \ref{app:hyperparameters} reports the selection, temperature, and rare-type analyses.
\begin{table}[htbp]
\centering
\small
\setlength{\tabcolsep}{7pt}
\caption{Normal-only anchor selection across mixtures of $C$ machine types. Counts show mode / range over five seeds; AUC entries are five-seed means (\%). The test oracle is used only for retrospective comparison.}
\label{tab:selection_main}
\begin{tabular}{lccc}
\toprule
Groups & Selected $K$ & Selected AUC & Oracle gap \\
\midrule
2 types & 4 / 2--6 & 92.84 & 0.27 \\
3 types & 6 / 4--8 & 91.97 & 0.32 \\
4 types & 8 / 6--12 & 91.31 & 0.27 \\
5 types & 8 / 8--12 & 90.64 & 0.37 \\
6 types & 8 / 8--12 & 90.07 & 0.19 \\
\bottomrule
\end{tabular}
\end{table}

\paragraph{Recorded noise and contaminated training.}
The additional tests probe distinct departures from the standard setting. With DEMAND background noise at $-6$ dB, HARMONY reaches 79.42\% AUC against 66.11\% for Audio-PC; under within-clip varying SNR, the corresponding AUCs are 83.12\% and 71.84\%. In the separate contamination protocol, every potential training contaminant is removed from the test set in all conditions. HARMONY decreases from 89.65\% AUC at 0\% contamination to 83.31\% at 5\%. Full results and protocols are in Appendices \ref{app:stress_tests} and \ref{app:contamination}; contamination still reduces performance, so the clean-training assumption remains relevant.

\paragraph{Acoustic interpretation.}
Post-hoc machine-type alignment complements the geometric evidence. On DCASE 2020, HARMONY's anchor assignments achieve NMI 0.78 and purity 0.89, compared with 0.42 and 0.63 for raw-feature KMeans. On MIMII, the corresponding values are 0.72 and 0.85. These labels are used only for analysis, not training or selection. Appendix \ref{app:tsne_visualization} provides the visualization and complete alignment measures. The anchors therefore capture structure related to machine acoustics while allowing multiple regions per machine type.

\subsection{RQ4: Does the Learned Partition Also Improve Retrieval Efficiency?}
\label{sec:rq4}

Finally, we ask whether the same learned partition that improves accuracy can also make inference cheaper. To isolate this question cleanly, we compare HARMONY's exhaustive retrieval and hierarchical retrieval on the same CPU and with the same backbone.

\begin{table}[ht]
\normalsize
\renewcommand{\arraystretch}{1.15}
\setlength{\tabcolsep}{12pt} 
\centering
\caption{Efficiency versus accuracy trade-off on CPU. Hierarchical retrieval (IVF) offers substantial latency reduction (\textcolor{posgreen}{green}) with negligible accuracy drop (\textcolor{negred}{red}) compared to exhaustive search (Flat).}
\label{tab:efficiency}
\begin{tabular}{l cccc}
\toprule
\multirow{2}{*}{Retrieval Method} & \multicolumn{2}{c}{Accuracy (\%)} & \multicolumn{2}{c}{Efficiency (CPU)} \\
\cmidrule(lr){2-3} \cmidrule(lr){4-5}
& AUC ($\uparrow$) & pAUC ($\uparrow$) & Latency ($\downarrow$) & Speedup ($\uparrow$) \\
\midrule
HARMONY (Flat) & \textbf{90.00} & \textbf{74.90} & 42.3 ms & 1.0$\times$ \\
\rowcolor{harmonyblue}
HARMONY (IVF)  & 89.80\negc{0.20} & 74.70\negc{0.20} & \textbf{5.4 ms} {\scriptsize\textcolor{posgreen}{(-87.2\%)}} & \textbf{7.8$\times$} \\
\bottomrule
\end{tabular}
\end{table}

\cref{tab:efficiency} shows that hierarchical retrieval reduces latency from 42.3 ms to 5.4 ms on the same CPU, a 7.8$\times$ speedup with a 0.20-point AUC drop. The comparison uses the same learned representation, isolating the effect of hierarchical candidate pruning.

\section{Conclusion}
\label{sec:conclusion}
HARMONY addresses metadata-free machinery ASD by reorganizing mixed normal features into anchor-defined regions and reusing the same geometry for hierarchical retrieval. On DCASE 2020, flat retrieval reaches 90.00\% mean AUC; hierarchical retrieval reduces CPU latency from 42.3 to 5.4 ms with a 0.20-point AUC reduction. Comparisons on MIMII and DCASE 2022 further support the approach across industrial acoustic conditions.

\textbf{Limitations}\quad Performance depends on the initial pre-trained features, and a global temperature can be suboptimal for strongly imbalanced modes. The separation and retrieval bounds require bounded residuals and candidate coverage, respectively. Density-aware partitioning and validation beyond industrial machinery remain future work. Appendix \ref{app:limitations_impacts} discusses broader impacts.

\clearpage

\begin{ack}
The authors received no third-party funding or support for this work. The authors declare no competing interests.
\end{ack}

\bibliography{references}
\bibliographystyle{plainnat}

\newpage
\section*{NeurIPS Paper Checklist}
\begin{enumerate}

\item {\bf Claims}
    \item[] Question: Do the main claims made in the abstract and introduction accurately reflect the paper's contributions and scope?
    \item[] Answer: \answerYes{} 
    \item[] Justification: The main claims regarding accuracy improvements and retrieval efficiency speedups are clearly stated in the Abstract and \cref{sec:intro}, and are empirically validated in \cref{sec:experiments}.

\item {\bf Limitations}
    \item[] Question: Does the paper discuss the limitations of the work performed by the authors?
    \item[] Answer: \answerYes{} 
    \item[] Justification: We explicitly discuss the limitations of our method, including vulnerabilities to extreme domain shifts and the rigidity of uniform temperature scaling, in \cref{app:limitations_impacts}.

\item {\bf Theory assumptions and proofs}
    \item[] Question: For each theoretical result, does the paper provide the full set of assumptions and a complete (and correct) proof?
    \item[] Answer: \answerYes{} 
    \item[] Justification: The geometric results state their residual and candidate-coverage conditions in Section \ref{sec:method}; their proofs and the distinct optimization assumptions are given in Appendix \ref{app:theory}.

    \item {\bf Experimental result reproducibility}
    \item[] Question: Does the paper fully disclose all the information needed to reproduce the main experimental results of the paper to the extent that it affects the main claims and/or conclusions of the paper (regardless of whether the code and data are provided or not)?
    \item[] Answer: \answerYes{} 
    \item[] Justification: We provide detailed experimental setups, network architectures, and baseline adaptation details in \cref{sec:experiments} and \cref{app:implementation}, with the public code repository linked in the Abstract.

\item {\bf Open access to data and code}
    \item[] Question: Does the paper provide open access to the data and code, with sufficient instructions to faithfully reproduce the main experimental results, as described in supplemental material?
    \item[] Answer: \answerYes{} 
    \item[] Justification: The public code repository is linked in the Abstract, and dataset and implementation details are given in Appendices \ref{app:datasets} and \ref{app:implementation}.

\item {\bf Experimental setting/details}
    \item[] Question: Does the paper specify all the training and test details (e.g., data splits, hyperparameters, how they were chosen, type of optimizer) necessary to understand the results?
    \item[] Answer: \answerYes{} 
    \item[] Justification: Comprehensive training configurations, data splits, and hyperparameters are concisely presented in \cref{sec:experiments} and comprehensively expanded in \cref{app:implementation}.

\item {\bf Experiment statistical significance}
    \item[] Question: Does the paper report error bars suitably and correctly defined or other appropriate information about the statistical significance of the experiments?
    \item[] Answer: \answerYes{} 
    \item[] Justification: We report the mean and standard deviation over 5 random seeds to demonstrate statistical significance and model stability, which is fully detailed in \cref{app:multi_seed}.

\item {\bf Experiments compute resources}
    \item[] Question: For each experiment, does the paper provide sufficient information on the computer resources (type of compute workers, memory, time of execution) needed to reproduce the experiments?
    \item[] Answer: \answerYes{} 
    \item[] Justification: The exact hardware used for training and inference latency measurements is explicitly stated in \cref{app:implementation}.
    
\item {\bf Code of ethics}
    \item[] Question: Does the research conducted in the paper conform, in every respect, with the NeurIPS Code of Ethics \url{https://neurips.cc/public/EthicsGuidelines}?
    \item[] Answer: \answerYes{} 
    \item[] Justification: The research strictly conforms to the NeurIPS Code of Ethics, utilizing standard, publicly available acoustic anomaly detection datasets without privacy violations.

\item {\bf Broader impacts}
    \item[] Question: Does the paper discuss both potential positive societal impacts and negative societal impacts of the work performed?
    \item[] Answer: \answerYes{} 
    \item[] Justification: We discuss both the positive impacts (enhancing industrial safety) and potential negative impacts (privacy risks from acoustic monitoring) in \cref{app:limitations_impacts}.
    
\item {\bf Safeguards}
    \item[] Question: Does the paper describe safeguards that have been put in place for responsible release of data or models that have a high risk for misuse (e.g., pre-trained language models, image generators, or scraped datasets)?
    \item[] Answer: \answerNA{} 
    \item[] Justification: Our work addresses acoustic anomaly detection for machinery and does not release high-risk generative models or personally identifiable scraped datasets.

\item {\bf Licenses for existing assets}
    \item[] Question: Are the creators or original owners of assets (e.g., code, data, models), used in the paper, properly credited and are the license and terms of use explicitly mentioned and properly respected?
    \item[] Answer: \answerYes{} 
    \item[] Justification: All existing assets, including DCASE 2020, MIMII, DCASE 2022, DEMAND, and the BYOL-A backbone, are properly cited and publicly available for research purposes.

\item {\bf New assets}
    \item[] Question: Are new assets introduced in the paper well documented and is the documentation provided alongside the assets?
    \item[] Answer: \answerNA{} 
    \item[] Justification: We do not introduce new datasets or foundational models; our proposed algorithmic methodology and corresponding code are documented in the provided repository.

\item {\bf Crowdsourcing and research with human subjects}
    \item[] Question: For crowdsourcing experiments and research with human subjects, does the paper include the full text of instructions given to participants and screenshots, if applicable, as well as details about compensation (if any)? 
    \item[] Answer: \answerNA{} 
    \item[] Justification: This research solely utilizes existing industrial machine sound datasets and does not involve any human subjects or crowdsourcing.

\item {\bf Institutional review board (IRB) approvals or equivalent for research with human subjects}
    \item[] Question: Does the paper describe potential risks incurred by study participants, whether such risks were disclosed to the subjects, and whether Institutional Review Board (IRB) approvals (or an equivalent approval/review based on the requirements of your country or institution) were obtained?
    \item[] Answer: \answerNA{} 
    \item[] Justification: The research does not involve human subjects, making IRB approval inapplicable.

\item {\bf Declaration of LLM usage}
    \item[] Question: Does the paper describe the usage of LLMs if it is an important, original, or non-standard component of the core methods in this research? Note that if the LLM is used only for writing, editing, or formatting purposes and does \emph{not} impact the core methodology, scientific rigor, or originality of the research, declaration is not required.
    \item[] Answer: \answerNA{} 
    \item[] Justification: Large Language Models (LLMs) were not used as a component of the core methods in this research.

\end{enumerate}

\clearpage
\appendix
\onecolumn

\FloatBarrier
\section{Comprehensive Notation Table}
\label{app:notation}

To facilitate the review of our theoretical framework and empirical evaluation, we summarize the key mathematical notations used throughout this paper in \cref{tab:notations}.

\begin{table}[ht]
\normalsize
\renewcommand{\arraystretch}{1.15}
\setlength{\tabcolsep}{8pt}
\caption{Summary of key mathematical notations.}
\label{tab:notations}
\centering
\begin{tabular}{@{}ll@{}}
\toprule
\textbf{Symbol} & \textbf{Definition} \\
\midrule
$\phi : \mathcal{X} \rightarrow \mathcal{Z}$ & Pre-trained frozen encoder mapping audio to feature space. \\
$L$ / $K$ & Latent normal-mode count / learned anchor count (not necessarily equal). \\
$D$ & Ambient dimension of the full input feature space $\mathcal{Z}$. \\
$\phi_{\#} \mathcal{P}$ & Pushforward measure of probability distribution $\mathcal{P}$ under mapping $\phi$. \\
$g_\theta : \mathcal{Z} \rightarrow \mathcal{H}$ & Trainable projection network mapping features to hypersphere $\mathbb{S}^{d-1}$. \\
$d'$ / $\eta$ & Intrinsic dimensionality / Cumulative energy threshold for spectral filtering. \\
$\mathcal{A} = \{a_k\}_{k=1}^K$ & Set of $K$ anchor points (topological skeleton) on the hypersphere. \\
$p(k | z_i; \theta)$ & Boltzmann probability of assigning sample $z_i$ to anchor $a_k$. \\
$\tau$ & Temperature controlling the sharpness of Boltzmann soft assignment. \\
$y^*_i$ & Pseudo-label (anchor index) with the highest assignment probability. \\
$\mathcal{L}_{align}(\theta)$ & Alignment loss: negative log-likelihood of pseudo-labels. \\
$\xi$ & Bounded residual alignment loss at convergence. \\
$\Delta_{\min}$ & Minimum Euclidean distance between distinct anchors. \\
$\epsilon$ & Quantization radius; max distance of local buckets for hierarchical retrieval. \\
$\mathcal{H}_{norm}$ & The projected normal database in the latent space. \\
$s(q)$ / $\hat{s}(q)$ & Ideal exhaustive-search score / Practical hierarchical retrieval score. \\
\bottomrule
\end{tabular}
\end{table}
\FloatBarrier
\section{Theoretical Proofs and Analysis}
\label{app:theory}

This section provides the rigorous mathematical derivations for the geometric properties and retrieval bounds presented in \cref{sec:method}. We further detail the asymptotic complexity analysis to substantiate the efficiency gains reported in \cref{sec:experiments}.

\subsection{Derivation of the gradient decomposition}
\label{deriva_lemma4.1}
In the anchor-guided manifold disentanglement stage, we optimize the alignment loss $\mathcal{L}_{\text{align}}(\theta) = \frac{1}{N} \sum_{i=1}^{N} l_i(z_i; \theta)$. For a given sample $z_i$, let $u_i = g_\theta(z_i) \in \mathbb{S}^{d-1}$ denote the projected feature on the hypersphere. The alignment loss for a single sample is defined as the negative log-likelihood of the pseudo-labeled anchor $y_i^* = \arg \max_k p(k|z_i; \theta)$:
\begin{equation}
l_i = -\log p(y_i^*|z_i; \theta) = -\log \left( \frac{\exp(\langle u_i, a_{y_i^*} \rangle / \tau)}{\sum_{j=1}^{K} \exp(\langle u_i, a_j \rangle / \tau)} \right),
\end{equation}
where $\tau$ is the temperature scaling parameter. To understand the geometric forces acting on the manifold, we compute the gradient of $l_i$ with respect to the projected feature $u_i$. Let $f_k = \langle u_i, a_k \rangle / \tau$ denote the scaled logit for the $k$-th anchor. By applying the chain rule, the gradient is expressed as:
\begin{equation}
\nabla_{u_i} l_i = \sum_{k=1}^{K} \frac{\partial l_i}{\partial f_k} \frac{\partial f_k}{\partial u_i}.
\end{equation}

First, to compute the partial derivative $\frac{\partial l_i}{\partial f_k}$, we consider two distinct cases for the index $k$:

\paragraph{Case A: $k = y_i^*$ (Target Anchor)}
When $k$ corresponds to the assigned pseudo-label $y_i^*$, the derivative is calculated as:
\begin{align}
\frac{\partial l_i}{\partial f_{y_i^*}} &= \frac{\partial}{\partial f_{y_i^*}} \left[ -f_{y_i^*} + \log \left( \sum_{j=1}^K \exp(f_j) \right) \right] \nonumber \\
&= -1 + \frac{1}{\sum_{j=1}^K \exp(f_j)} \cdot \frac{\partial}{\partial f_{y_i^*}} \left( \sum_{j=1}^K \exp(f_j) \right) \nonumber \\
&= -1 + \frac{\exp(f_{y_i^*})}{\sum_{j=1}^K \exp(f_j)} \nonumber \\
&= p(y_i^*|z_i) - 1.
\end{align}

\paragraph{Case B: $k \neq y_i^*$ (Non-target Anchors)}
When $k$ is any other anchor index, the term $-f_{y_i^*}$ is treated as a constant with respect to $f_k$:
\begin{align}
\frac{\partial l_i}{\partial f_k} &= \frac{\partial}{\partial f_k} \left[ -f_{y_i^*} + \log \left( \sum_{j=1}^K \exp(f_j) \right) \right] \nonumber \\
&= 0 + \frac{1}{\sum_{j=1}^K \exp(f_j)} \cdot \frac{\partial}{\partial f_k} \left( \sum_{j=1}^K \exp(f_j) \right) \nonumber \\
&= \frac{\exp(f_k)}{\sum_{j=1}^K \exp(f_j)} \nonumber \\
&= p(k|z_i).
\end{align}

\paragraph{Unified Expression}
Combining both cases, we obtain the general form for the first term:
\begin{equation}
\frac{\partial l_i}{\partial f_k} = p(k|z_i) - \mathbb{I}(k = y_i^*),
\end{equation}
where $\mathbb{I}(\cdot)$ is the indicator function. 

Second, the derivative of the logit $f_k$ with respect to the projected feature $u_i$ is simply the scaled anchor vector:
\begin{equation}
\frac{\partial f_k}{\partial u_i} = \frac{1}{\tau} a_k.
\end{equation}

Substituting these components back into the chain rule summation, we obtain the total gradient:
\begin{align}
\nabla_{u_i} l_i &= \sum_{k=1}^{K} \left[ p(k|z_i) - \mathbb{I}(k = y_i^*) \right] \frac{1}{\tau} a_k \nonumber \\
&= \frac{1}{\tau} \sum_{j=1}^{K} p(j|z_i) a_j - \frac{1}{\tau} a_{y_i^*}.
\end{align}

Rearranging the terms yields the geometric force formulation presented in the main text:
\begin{equation}
\nabla_{u_i} l_i = -\frac{1}{\tau} a_{y_i^*} + \frac{1}{\tau} \sum_{j=1}^{K} p(j|z_i) a_j. \label{eq:grad_8}
\end{equation}

As stated in Lemma~\ref{lem:gradient}, the negative gradient $-\nabla_{u_i} l_i$ represents the direction of parameter updates, which can be decomposed into an attraction force toward the assigned anchor $a_{y_i^*}$ and a repulsion force away from competing anchors $a_j$.

\subsection{Optimization Assumptions and Residual Bound}
\label{app:opt_residual}

\begin{assumption}[L-Smoothness]
\label{ass:smooth}
The alignment loss $\mathcal{L}_{align}(\theta)$ is $L$-smooth with respect to $\theta$, i.e., $\|\nabla \mathcal{L}(\theta)-\nabla \mathcal{L}(\theta')\| \le L \|\theta-\theta'\|$.
\end{assumption}

\begin{proposition}[Loss Descent]
\label{prop:convergence}
Under Assumption~\ref{ass:smooth}, gradient descent with step size $\eta \le 1/L$ guarantees monotonic descent:
\begin{equation}
    \mathcal{L}_{align}(\theta_{t+1})
    \le
    \mathcal{L}_{align}(\theta_t)
    - \frac{\eta}{2}\|\nabla \mathcal{L}_{align}(\theta_t)\|^2.
\end{equation}
If the loss is bounded below, summing the descent inequality bounds the sum of squared gradient norms. This descent property does not itself establish a prescribed loss level or a uniform sample-to-anchor residual.
\end{proposition}

\paragraph{Residual condition.}
The separation result in the main text assumes a uniform training residual $\delta_{train}$ and an additional test deviation $\eta_{gen}$. A local loss--distance relation of the form $\|g_\theta(z_i)-a_{y_i^*}\|_2\leq C_\tau\sqrt{\ell_i}$ requires a separate local error-bound condition. It does not follow from smoothness or a small mean cross-entropy alone. We therefore state the geometric proof directly in terms of the residual assumed by Theorem~\ref{thm:separation}.

\subsection{Proof of the manifold separation bound}

\begin{theorem}[Restatement of Theorem \ref{thm:separation}]
Let $\mathcal{A} = \{a_1, \dots, a_K\}$ be fixed spectral anchors with minimum Euclidean separation $\Delta_{\min} = \min_{i \neq j} \|a_i - a_j\|_2$. Suppose every normal sample considered is within $\delta_{train}+\eta_{gen}$ of its assigned anchor. For two samples $x,x'$ assigned to distinct anchors,
\begin{equation}
\|g_\theta(x)-g_\theta(x')\|_2\geq\Delta_{\min}-2(\delta_{train}+\eta_{gen}).
\end{equation}
\end{theorem}

\begin{proof}
\textbf{1. Decomposition of Inter-Sample Distance.}
Let $x$ be a sample assigned to anchor $a_k$ (i.e., $y^*(x)=k$) and $x'$ be a sample assigned to anchor $a_j$ (i.e., $y^*(x')=j$), with $k \neq j$.
We seek to bound the distance $\|g_\theta(x) - g_\theta(x')\|_2$. By the \textit{Reverse Triangle Inequality}, for any vectors $u, v, w$, we have $\|u - v\|_2 \ge \|u - w\|_2 - \|w - v\|_2$.
Applying this to our geometric configuration by introducing the anchor points $a_k$ and $a_j$ as intermediaries:
\begin{align}
    \label{eq:triangle_decomp}
    \|g_\theta(x) - g_\theta(x')\|_2 &= \|(g_\theta(x) - a_k) + (a_k - a_j) + (a_j - g_\theta(x'))\|_2 \\
    &\ge \|a_k - a_j\|_2 - \|g_\theta(x) - a_k\|_2 - \|g_\theta(x') - a_j\|_2.
\end{align}

\textbf{2. Bounding the Terms.}
\begin{itemize}
    \item \textbf{Anchor Separation:} By definition of the spectral anchors (which are fixed via stop-gradient), the distance between distinct centroids is lower-bounded by $\Delta_{\min}$:
    \begin{equation}
        \|a_k - a_j\|_2 \ge \Delta_{\min}.
    \end{equation}
    
    \item \textbf{Within-anchor residuals:} By the theorem's residual assumption,
    $\|g_\theta(x)-a_k\|_2\leq\delta_{train}+\eta_{gen}$ and
    $\|g_\theta(x')-a_j\|_2\leq\delta_{train}+\eta_{gen}$.
    The test deviation is an assumption for the samples covered by the bound, not a uniform guarantee derived from Lipschitz continuity alone.
\end{itemize}

\textbf{3. Synthesis.}
Substituting these bounds into Eq. \ref{eq:triangle_decomp}:
\begin{equation}
    \|g_\theta(x) - g_\theta(x')\|_2 \ge \Delta_{\min} - (\delta_{train} + \eta_{gen}) - (\delta_{train} + \eta_{gen}).
\end{equation}
Simplifying yields the margin lower bound:
\begin{equation}
    \|g_\theta(x) - g_\theta(x')\|_2 \ge \Delta_{\min} - 2(\delta_{train} + \eta_{gen}).
    \label{eq:derived_margin} 
\end{equation}
This establishes separation between different assigned regions under the stated residual condition. It does not require a one-to-one correspondence between anchors and machine identities.
\end{proof}

\begin{remark}[Non-vacuousness]
The lower bound is positive when $\Delta_{\min}>2(\delta_{train}+\eta_{gen})$. The mean cosine-distance proxies in \cref{tab:qe_margin} describe the observed geometry, but are not estimates of the uniform Euclidean residual or the minimum anchor distance in this theorem.
\end{remark}

\subsection{Proof of the retrieval error bound}

\begin{theorem}[Restatement of Theorem \ref{thm:robustness}]
Let $s(q) = \min_{h \in \mathcal{H}_{norm}} \|q - h\|_2$ be the exact nearest neighbor distance, and $\hat{s}(q)$ be the distance computed within the retrieved subset $\mathcal{S} \subset \mathcal{H}_{norm}$. Suppose the candidate set contains a point $\hat h$ sharing a bucket centroid $c$ with a true nearest neighbor $h^*$, with both within radius $\epsilon$ of $c$. Then:
\begin{equation}
    |\hat{s}(q) - s(q)| \le 2\epsilon.
\end{equation}
\end{theorem}

\begin{proof}
\textbf{1. Upper Bound (Safety).} 
Since the retrieved subset $\mathcal{S}$ is a subset of the full projected database $\mathcal{H}_{norm}$ ($\mathcal{S} \subset \mathcal{H}_{norm}$), the minimum distance found in $\mathcal{S}$ can only be greater than or equal to the global minimum:
\begin{equation}
    \hat{s}(q) \ge s(q).
\end{equation}

\textbf{2. Lower Bound (Coverage).} 
Our hierarchical index partitions the space based on the Voronoi topology derived from anchors $\mathcal{A}$ and bucket centroids $\mathcal{B}$.
For a query $q$, let $h^*$ be its true nearest neighbor in $\mathcal{H}_{norm}$. The coverage condition requires a retrieved candidate in the same quantization bucket as $h^*$. Top-$K_1$ expert selection alone does not establish this condition. If all points of that bucket are retrieved, $h^*$ itself is a candidate and the score error is zero.
Within this bucket, vector quantization guarantees that every sample $h$ is within distance $\epsilon$ of its bucket centroid $c$: $\|h - c\|_2 \le \epsilon$.
Let $\hat{h} \in \mathcal{S}$ be the candidate sample closest to $q$ retrieved from the same bucket as $h^*$. By the triangle inequality, the distance between the true neighbor and the retrieved candidate is bounded by:
\begin{equation}
    \|h^* - \hat{h}\|_2 \le \|h^* - c\|_2 + \|c - \hat{h}\|_2 \le 2\epsilon.
\end{equation}
Applying the triangle inequality again for the query distance:
\begin{equation}
    \|q - \hat{h}\|_2 \le \|q - h^*\|_2 + \|h^* - \hat{h}\|_2 \le s(q) + 2\epsilon.
\end{equation}
Since $\hat{s}(q) = \min_{h \in \mathcal{S}} \|q - h\|_2$, we have $\hat{s}(q) \le \|q - \hat{h}\|_2 \le s(q) + 2\epsilon$.

\textbf{3. Conclusion.}
Combining the bounds, we have $s(q) \le \hat{s}(q) \le s(q) + 2\epsilon$. Thus, the pointwise score distortion is rigorously bounded by $2\epsilon$.
\end{proof}

\subsection{Complexity Analysis of Hierarchical Retrieval}

We formally derive the time complexity of the proposed Expert-Bucket Indexing.

\paragraph{1. Brute-Force Search (Flat).}
Standard anomaly detection (e.g., PatchCore) computes distances to all $N$ samples in the memory bank.
\begin{equation}
    T_{\text{Flat}} = \mathcal{O}(N \cdot d),
\end{equation}
where $d$ is the projected feature dimension.

\paragraph{2. Hierarchical Routing.}
Our method employs a two-stage pruning:
\begin{itemize}
    \item \textbf{Stage 1 (Expert Routing):} Compute distances to $K$ spectral anchors. Cost: $\mathcal{O}(K \cdot d)$. Select top-$K_1$ experts.
    \item \textbf{Stage 2 (Bucket Probing):} Within selected experts, compute distances to $B$ local centroids. Cost: $\mathcal{O}(K_1 \cdot B \cdot d)$. Select top-$K_2$ buckets.
    \item \textbf{Stage 3 (Candidate Scoring):} Compute exact distances only to samples in retrieved buckets. Assuming uniform distribution, the number of candidates is $N' \approx \frac{N}{K \cdot B} \cdot K_1 \cdot K_2$. Cost: $\mathcal{O}(N' \cdot d)$.
\end{itemize}

\paragraph{3. Total Complexity \& Speedup.}
The total complexity is:
\begin{equation}
    T_{\text{Ours}} = \mathcal{O}\left( \left(K + K_1 B + \frac{K_1 K_2 N}{K B}\right) \cdot d \right).
    \label{eq:complexity_final}
\end{equation}
Since $K, B, K_1, K_2$ are small constants (e.g., $K=8, B=64$), and the dimension $d$ is fixed, the term $\frac{N}{K \cdot B}$ dominates the scaling. While this remains asymptotically $\mathcal{O}(N)$, it achieves a significant constant-factor reduction in linear scalability relative to the raw dataset size.
\FloatBarrier
\section{Empirical Analysis of Learned Geometry}
\label{app:empirical_theory}

To characterize the learned representation empirically, we provide quantitative proxies and statistical analyses extracted from the optimization process.

\subsection{Anchor Occupancy and Stability}
Theoretical proofs require that the pseudo-assignment mechanism does not collapse into a trivial solution. We extracted the unedited anchor occupancy statistics across varying machine complexities and anchor counts. Across all configurations, the number of empty anchors is strictly zero, and the Shannon Entropy remains consistently near the theoretical maximum, indicating broad utilization in these evaluated mixtures, as visually depicted in \cref{fig:anchor_occupancy}.

\begin{figure}[ht]
    \centering
    \includegraphics[width=\linewidth]{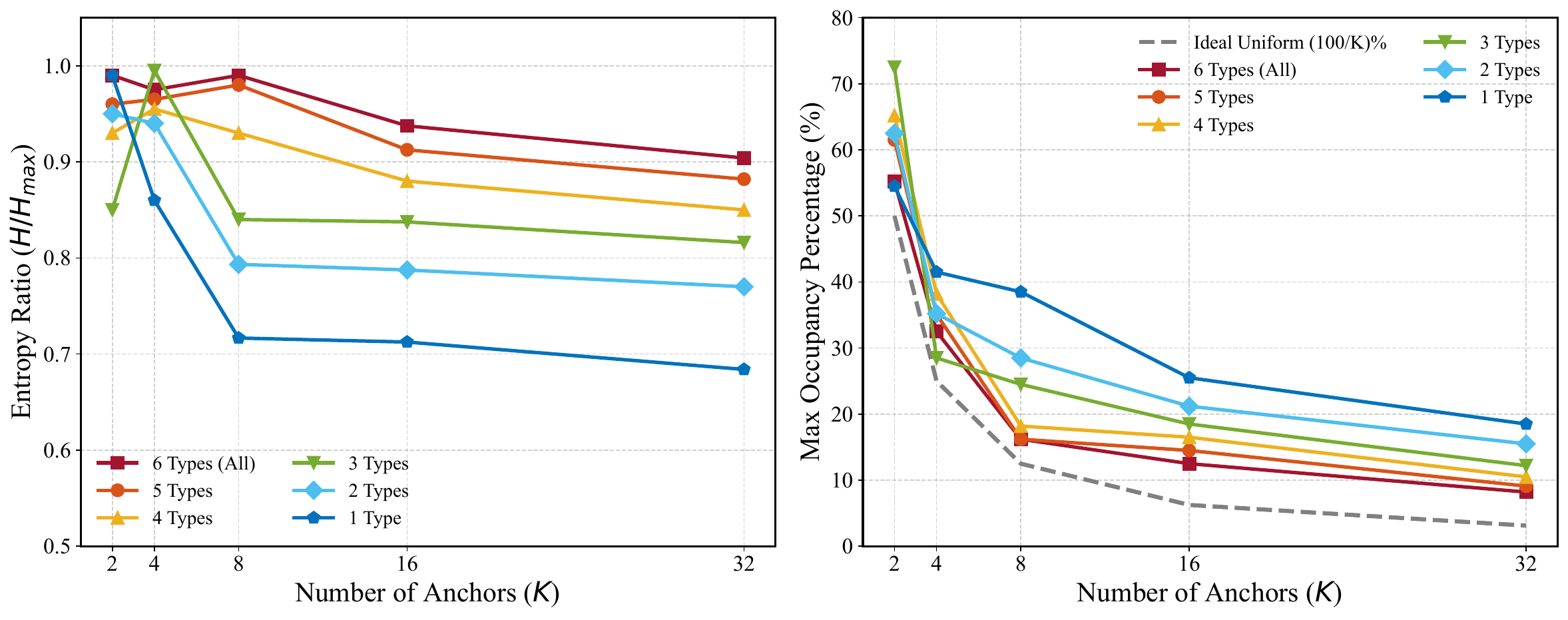}
    \caption{Anchor occupancy statistics across varying machine complexities. (Left) The Entropy Ratio remains near 1.0, indicating uniform utilization. (Right) The maximum occupancy percentage tightly tracks the ideal uniform distribution boundary, showing nonempty anchors in these evaluated configurations.}
    \label{fig:anchor_occupancy}
\end{figure}

\subsection{Quantization Error and Geometric Margin}
We measure empirical proxies: the Quantization Error (QE) proxy (mean distance from samples to assigned anchors) and the Margin proxy (mean distance difference between the first and second nearest anchors). As shown in \cref{tab:qe_margin}, the projected space has lower quantization error, a wider margin proxy, and higher detection performance than the raw space. Specifically, HARMONY improves the relative separation (Margin/QE Ratio) by nearly 10$\times$ compared to the raw feature space, confirming the effectiveness of active geometric disentanglement.
The geometric comparison is summarized in the main text in \cref{tab:qe_margin}.

\subsection{Optimization Smoothness}
\cref{prop:convergence} states a descent inequality under a smoothness and step-size condition. We monitored the loss trajectory during training, showing a decreasing empirical objective. \cref{fig:loss_curve} illustrates this smooth optimization trajectory over 20 epochs, showing no severe oscillations or divergence in the observed run.

\begin{figure}[ht]
    \centering
    \includegraphics[width=0.6\linewidth]{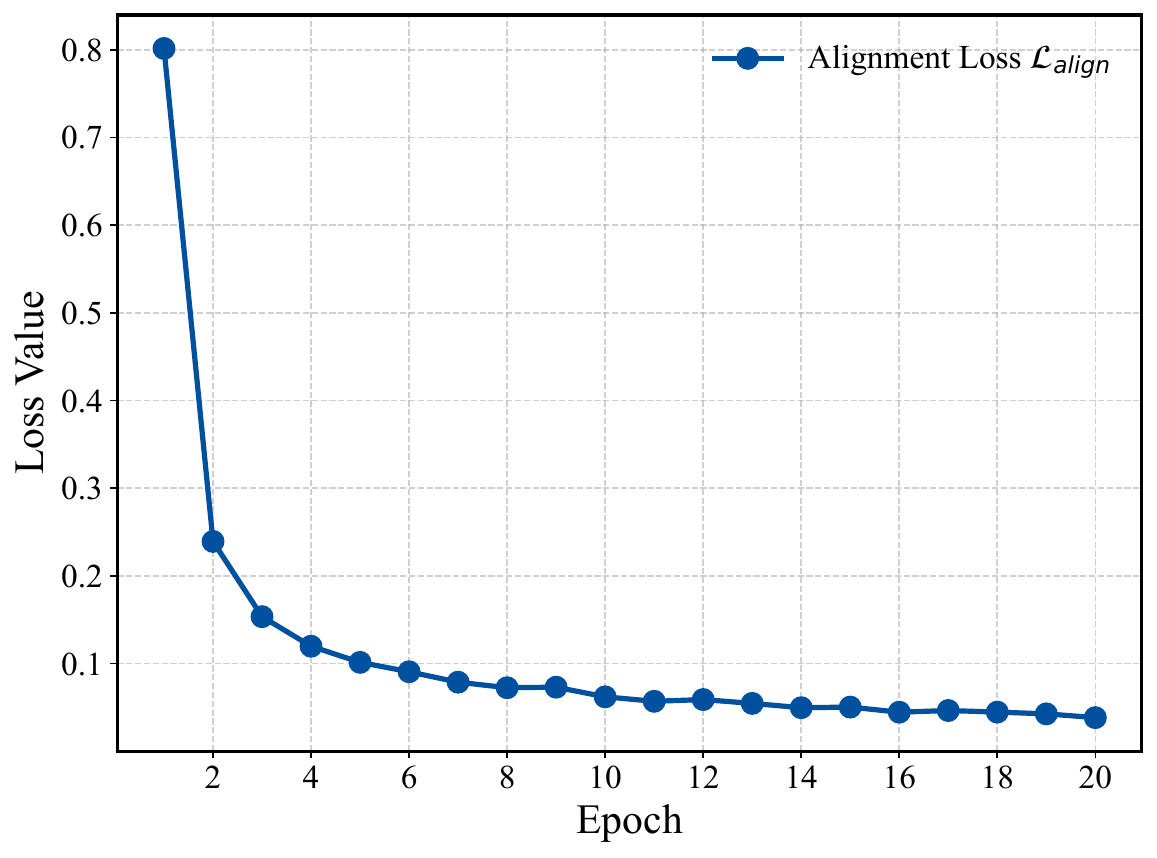}
    \caption{Smooth optimization trajectory of the alignment loss $\mathcal{L}_{align}$ over 20 epochs, showing the observed loss trajectory.}
    \label{fig:loss_curve}
\end{figure}

\subsection{Anomaly Score Distribution}
Our method relies on anomalies mapping to low-density regions far from learned anchor centers. \cref{tab:anomaly_dist} details the distribution of $s(q)$ for test queries. Anomalous samples exhibit a median score approximately 15$\times$ larger than normal queries, showing score separation in this test sample.

\begin{table}[ht]
\normalsize
\renewcommand{\arraystretch}{1.15}
\setlength{\tabcolsep}{10pt}
\caption{Distribution of Anomaly Scores $s(q)$ for normal vs. anomalous test queries.}
\label{tab:anomaly_dist}
\centering
\begin{tabular}{l ccccc}
\toprule
\textbf{Group} & \textbf{Count} & \textbf{Mean} & \textbf{Median} & \textbf{95th Pctl} & \textbf{Min} \\
\midrule
Normal & 5,399 & 0.013 & 0.011 & 0.022 & 0.003 \\
Anomalous & 5,469 & 0.187 & 0.164 & 0.361 & 0.021 \\
\bottomrule
\end{tabular}
\end{table}
\FloatBarrier
\section{Visualization of Manifold Disentanglement}
\label{app:tsne_visualization}

To intuitively demonstrate the effectiveness of our active geometric repartitioning, we visualize the representation space before and after applying HARMONY. 

As discussed in Section \ref{sec:rq1}, the primary bottleneck in metadata-free machinery ASD is the entanglement of mixed-source normal data. \cref{fig:tsne_vis} provides qualitative evidence that our method effectively resolves this issue. The projection seamlessly restructures the chaotic, entangled raw feature space into highly compact, semantically consistent Voronoi regions, consistent with the observed geometric separation.

\begin{figure}[htbp] 
    \centering
    % 在附录中，如果单栏排版，0.7 到 0.8 的 linewidth 视觉效果最好
    \includegraphics[width=0.95\linewidth]{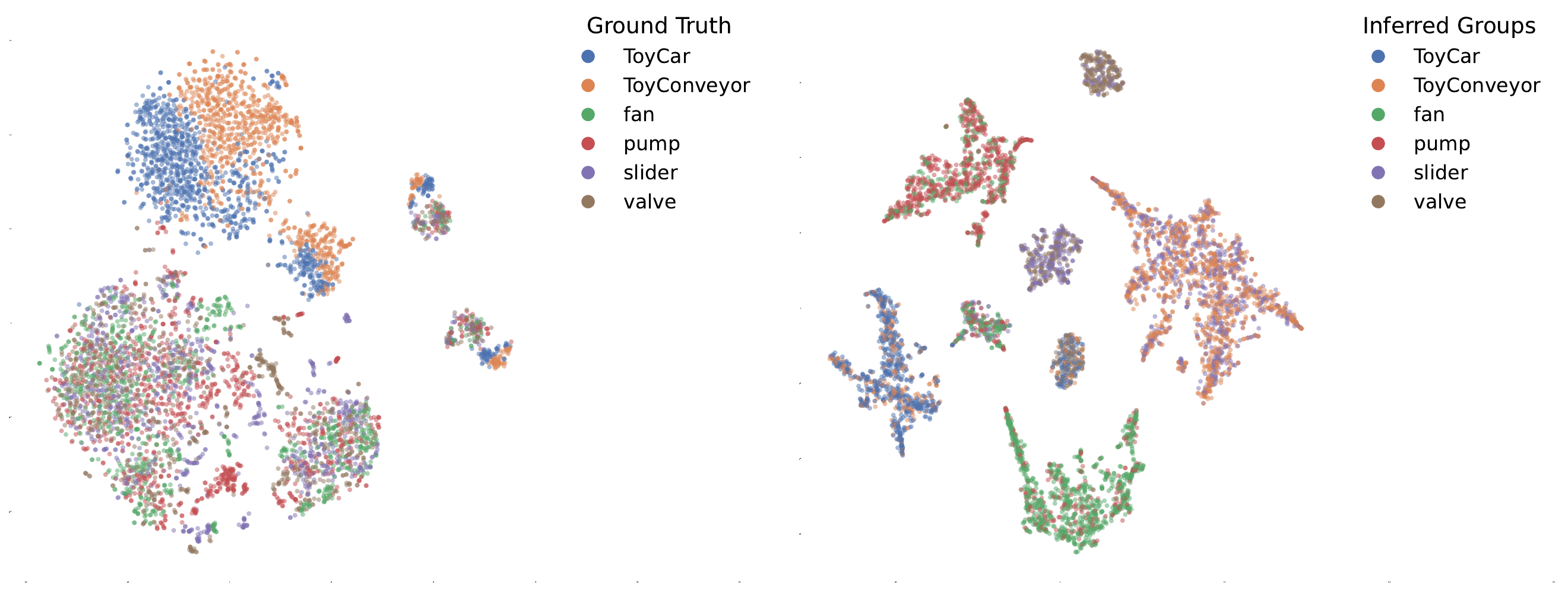}
    \caption{t-SNE \citep{maaten2008visualizing} visualization of manifold disentanglement. Projections transition from entangled raw features (Left) to separated Voronoi regions (Right), illustrating the learned feature organization.}
    \label{fig:tsne_vis}
\end{figure}
\subsection{Acoustic Meaning of the Anchors}
An anchor represents an acoustic mode rather than a fixed machine identity. To complement the visualization, \cref{tab:anchor_semantics} reports normalized mutual information (NMI), purity, homogeneity/completeness against machine types. Machine labels are used only for this post-hoc analysis. HARMONY improves the reported alignment measures on both datasets. Multiple anchors can describe different operating sounds from one machine type.
\begin{table}[htbp]
\centering
\small
\setlength{\tabcolsep}{4pt}
\caption{Post-hoc alignment with machine types. Hom./Comp. lists homogeneity/completeness; NMI and purity are mean $\pm$ standard deviation over five seeds.}
\label{tab:anchor_semantics}
\begin{tabular}{lccc}
\toprule
Dataset / method & NMI & Purity & Hom./Comp. \\
\midrule
DCASE / Raw KMeans & $0.42\pm0.03$ & $0.63\pm0.03$ & 0.49 / 0.38 \\
DCASE / MCDD & $0.66\pm0.03$ & $0.81\pm0.02$ & 0.74 / 0.57 \\
DCASE / HARMONY & $0.78\pm0.02$ & $0.89\pm0.01$ & 0.85 / 0.65 \\
MIMII / MCDD & $0.61\pm0.03$ & $0.77\pm0.03$ & 0.69 / 0.51 \\
MIMII / HARMONY & $0.72\pm0.03$ & $0.85\pm0.02$ & 0.80 / 0.59 \\
\bottomrule
\end{tabular}
\end{table}

\FloatBarrier
\section{Verification of Paradigm Innovation}
\label{app:verification}

To validate the universality and robustness of the proposed active disentanglement paradigm, we present component-level replacement experiments across multiple dimensions.

\subsection{Universality Across Backbones}
HARMONY is designed as a feature-agnostic plugin. \cref{tab:backbone} demonstrates that the framework significantly improves spatial decoupling and detection performance across various pre-trained feature extractors, including Masked SSL and Contrastive models.

\begin{table}[ht]
\normalsize
\renewcommand{\arraystretch}{1.15}
\setlength{\tabcolsep}{12pt}
\caption{Performance verification across different pre-trained backbones.}
\label{tab:backbone}
\centering
\begin{tabular}{llc}
\toprule
\textbf{Backbone} & \textbf{Model Type} & \textbf{Mean AUC (\%)} \\
\midrule
Audio-MAE  & Masked SSL & 92.60 \\
CLAP  & Contrastive & 87.40 \\
PANNS  & Supervised Semantic & 84.40 \\
\midrule
\rowcolor{harmonyblue}
BYOL-A (Default)  & Audio SSL & 90.00 \\
\bottomrule
\end{tabular}
\end{table}

\subsection{Robustness of Topological Initialization}
\cref{tab:init} verifies that while the framework benefits from the proposed Spectral + K-Means++ initialization, the performance remains highly stable even when substituting alternative parametric clustering methods, indicating robustness to the specific heuristic utilized.

\begin{table}[ht]
\normalsize
\renewcommand{\arraystretch}{1.15}
\setlength{\tabcolsep}{12pt}
\caption{Robustness test replacing the topological initialization method.}
\label{tab:init}
\centering
\begin{tabular}{llc}
\toprule
\textbf{Initialization Method} & \textbf{Assignment Method} & \textbf{Mean AUC (\%)} \\
\midrule
Spectral + GMM & Boltzmann & 89.40 \\
Spectral + vMF Mixture & Boltzmann & 90.70 \\
\rowcolor{harmonyblue}
Spectral + K-Means++ (Default) & Boltzmann & 90.00 \\
\bottomrule
\end{tabular}
\end{table}

\begin{table}[ht]
\normalsize
\renewcommand{\arraystretch}{1.15}
\setlength{\tabcolsep}{12pt}
\caption{Ablation of the assignment mechanism (Soft vs. Hard). Metrics are in percentages (\%).}
\label{tab:hard_soft}
\centering
\begin{tabular}{lcc}
\toprule
\textbf{Assignment Method} & \textbf{Mean AUC (\%)} & \textbf{$\Delta$ AUC} \\
\midrule
\rowcolor{harmonyblue}
Boltzmann (Ours)  & \textbf{90.00} & - \\
Hard Nearest  & 87.60 & \textcolor{negred}{-2.40} \\
\bottomrule
\end{tabular}
\end{table}

\subsection{Necessity of Push-Pull Assignment}
To isolate the contribution of the gradient-induced geometric forces, we replaced the Boltzmann soft assignment with hard labels (one-hot nearest neighbor). \cref{tab:hard_soft} shows a significant performance drop, confirming the necessity of soft probability matching for precise manifold reshaping.

\subsection{Efficiency vs. Accuracy Trade-off}
To further evaluate the practical deployment value of the proposed hierarchical retrieval strategy, we compare it against a standard exhaustive search (Flat). As shown in \cref{tab:efficiency_tradeoff}, our hierarchical index (IVF) utilizes the geometric margins established during training to achieve a 7.8$\times$ inference speedup on a standard CPU. This reduces the number of distance evaluations with a negligible performance drop (0.20 AUC points).

\begin{table}[ht]
\normalsize
\renewcommand{\arraystretch}{1.15}
\setlength{\tabcolsep}{12pt}
\caption{Efficiency vs. accuracy trade-off between retrieval strategies (measured on CPU).}
\label{tab:efficiency_tradeoff}
\centering
\begin{tabular}{lccc}
\toprule
\textbf{Retrieval Strategy} & \textbf{Mean AUC (\%)} & \textbf{Latency (CPU)} & \textbf{Speedup} \\
\midrule
Flat Search (Exact) & \textbf{90.00} & 42.3 ms & 1.0$\times$ \\
\rowcolor{harmonyblue}
Hierarchical Index (IVF) & 89.80 & \textbf{5.4 ms} & \textbf{7.8$\times$} \\
\bottomrule
\end{tabular}
\end{table}
\FloatBarrier
\section{Multi-Seed Stability and Detailed Metrics}
\label{app:multi_seed}

To thoroughly evaluate the robustness of our method against random initialization and statistical variance, we ran the main experiment across 5 different random seeds. 

\cref{tab:multi_seed_auc} and \cref{tab:multi_seed_pauc} report the detailed Mean and Standard Deviation (Mean $\pm$ Std) for the Area Under the Curve (AUC) and partial AUC (pAUC, $p=0.1$) on the DCASE 2020 dataset, respectively. The results confirm that HARMONY not only achieves state-of-the-art mean performance but also exhibits high stability (with standard deviations $\le 0.4\%$ across most machine types), demonstrating that the learned active geometric repartitioning is highly robust.

\begin{table*}[htbp]
\normalsize
\renewcommand{\arraystretch}{1.15}
\setlength{\tabcolsep}{3.5pt}
\caption{Multi-seed stability for universal anomaly detection (AUC, \%) on DCASE 2020. Metrics are reported as Mean $\pm$ Std over 5 random seeds.}
\label{tab:multi_seed_auc}
\begin{tabular}{l ccccccc}
\toprule
\multirow{2}{*}{Methods} & \multicolumn{6}{c}{Machine Types} & \multirow{2}{*}{Average} \\
\cmidrule(lr){2-7}
& ToyCar & ToyConv & Fan & Pump & Slider & Valve & \\
\midrule
PatchCore & 62.80{\footnotesize$\pm$0.1} & 60.60{\footnotesize$\pm$0.2} & 62.20{\footnotesize$\pm$0.1} & 57.20{\footnotesize$\pm$0.2} & 67.80{\footnotesize$\pm$0.2} & 61.50{\footnotesize$\pm$0.1} & 62.10{\footnotesize$\pm$0.2} \\
\rowcolor{groupgray}
\multicolumn{8}{c}{\textit{Classic Adaptation}} \\
SimpleNet   & 63.20{\footnotesize$\pm$0.7} & 62.10{\footnotesize$\pm$0.9} & 56.70{\footnotesize$\pm$0.6} & 58.00{\footnotesize$\pm$0.8} & 66.30{\footnotesize$\pm$0.7} & 62.20{\footnotesize$\pm$0.6} & 61.40{\footnotesize$\pm$0.7} \\
EfficientAD & 65.80{\footnotesize$\pm$0.4} & 65.60{\footnotesize$\pm$0.5} & 65.70{\footnotesize$\pm$0.3} & 68.50{\footnotesize$\pm$0.4} & 66.50{\footnotesize$\pm$0.3} & 65.60{\footnotesize$\pm$0.4} & 66.30{\footnotesize$\pm$0.4} \\
\rowcolor{groupgray}
\multicolumn{8}{c}{\textit{Foundation Embeddings}} \\
Gen-Rep     & 82.30{\footnotesize$\pm$0.2} & 81.10{\footnotesize$\pm$0.3} & 79.40{\footnotesize$\pm$0.2} & 79.80{\footnotesize$\pm$0.2} & 81.20{\footnotesize$\pm$0.1} & 79.90{\footnotesize$\pm$0.2} & 80.60{\footnotesize$\pm$0.2} \\
Audio-PC    & 85.40{\footnotesize$\pm$0.1} & 84.10{\footnotesize$\pm$0.2} & 82.50{\footnotesize$\pm$0.2} & 83.30{\footnotesize$\pm$0.2} & 85.10{\footnotesize$\pm$0.2} & 83.20{\footnotesize$\pm$0.1} & 83.90{\footnotesize$\pm$0.2} \\
\rowcolor{groupgray}
\multicolumn{8}{c}{\textit{Unified Models}} \\
CADA        & 73.30{\footnotesize$\pm$0.8} & 73.20{\footnotesize$\pm$0.9} & 78.30{\footnotesize$\pm$0.7} & 72.60{\footnotesize$\pm$0.9} & 73.10{\footnotesize$\pm$0.8} & 75.90{\footnotesize$\pm$0.7} & 74.50{\footnotesize$\pm$0.8} \\
UniFlow     & 75.20{\footnotesize$\pm$0.6} & 78.70{\footnotesize$\pm$0.5} & 75.40{\footnotesize$\pm$0.7} & 72.80{\footnotesize$\pm$0.6} & 75.30{\footnotesize$\pm$0.5} & 75.80{\footnotesize$\pm$0.7} & 75.40{\footnotesize$\pm$0.6} \\
Dinomaly    & 74.90{\footnotesize$\pm$0.5} & 73.80{\footnotesize$\pm$0.6} & 74.70{\footnotesize$\pm$0.4} & 76.00{\footnotesize$\pm$0.5} & 74.90{\footnotesize$\pm$0.5} & 76.80{\footnotesize$\pm$0.4} & 75.30{\footnotesize$\pm$0.5} \\
\rowcolor{groupgray}
\multicolumn{8}{c}{\textit{Disentanglement Baselines}} \\
GRHD        & 66.20{\footnotesize$\pm$0.4} & 63.90{\footnotesize$\pm$0.3} & 65.10{\footnotesize$\pm$0.4} & 63.40{\footnotesize$\pm$0.5} & 65.80{\footnotesize$\pm$0.3} & 65.30{\footnotesize$\pm$0.4} & 65.00{\footnotesize$\pm$0.4} \\
MCDD        & 79.10{\footnotesize$\pm$1.1} & 78.30{\footnotesize$\pm$0.8} & 76.90{\footnotesize$\pm$1.2} & 77.80{\footnotesize$\pm$0.9} & 79.40{\footnotesize$\pm$1.0} & 78.20{\footnotesize$\pm$0.8} & 78.30{\footnotesize$\pm$1.0} \\
AudDSR      & 81.20{\footnotesize$\pm$1.4} & 79.80{\footnotesize$\pm$1.2} & 78.50{\footnotesize$\pm$1.5} & 79.10{\footnotesize$\pm$1.3} & 81.30{\footnotesize$\pm$1.1} & 80.20{\footnotesize$\pm$1.4} & 80.00{\footnotesize$\pm$1.3} \\
\rowcolor{groupgray}
\multicolumn{8}{c}{\textit{Proposed Method}} \\
\rowcolor{harmonyblue}
HARMONY     & \textbf{98.00}{\footnotesize$\pm$0.2} & \textbf{89.70}{\footnotesize$\pm$0.3} & \textbf{86.20}{\footnotesize$\pm$0.4} & \textbf{88.70}{\footnotesize$\pm$0.3} & \textbf{89.70}{\footnotesize$\pm$0.2} & \textbf{87.60}{\footnotesize$\pm$0.3} & \textbf{90.00}{\footnotesize$\pm$0.3} \\
\bottomrule
\end{tabular}
\end{table*}

\begin{table*}[htbp]
\normalsize
\renewcommand{\arraystretch}{1.15}
\setlength{\tabcolsep}{3.5pt}
\caption{Multi-seed stability for universal anomaly detection (pAUC, \%) on DCASE 2020. Metrics are reported as Mean $\pm$ Std over 5 random seeds.}
\label{tab:multi_seed_pauc}
\begin{tabular}{l ccccccc}
\toprule
\multirow{2}{*}{Methods} & \multicolumn{6}{c}{Machine Types} & \multirow{2}{*}{Average} \\
\cmidrule(lr){2-7}
& ToyCar & ToyConv & Fan & Pump & Slider & Valve & \\
\midrule
PatchCore & 57.60{\footnotesize$\pm$0.1} & 56.00{\footnotesize$\pm$0.1} & 57.10{\footnotesize$\pm$0.2} & 52.30{\footnotesize$\pm$0.1} & 59.40{\footnotesize$\pm$0.1} & 54.90{\footnotesize$\pm$0.2} & 56.20{\footnotesize$\pm$0.1} \\
\rowcolor{groupgray}
\multicolumn{8}{c}{\textit{Classic Adaptation}} \\
SimpleNet   & 55.00{\footnotesize$\pm$0.6} & 53.20{\footnotesize$\pm$0.5} & 49.60{\footnotesize$\pm$0.4} & 52.90{\footnotesize$\pm$0.7} & 56.80{\footnotesize$\pm$0.6} & 56.00{\footnotesize$\pm$0.5} & 53.90{\footnotesize$\pm$0.5} \\
EfficientAD & 58.20{\footnotesize$\pm$0.3} & 57.40{\footnotesize$\pm$0.4} & 57.60{\footnotesize$\pm$0.2} & 59.30{\footnotesize$\pm$0.3} & 60.80{\footnotesize$\pm$0.2} & 58.30{\footnotesize$\pm$0.3} & 58.70{\footnotesize$\pm$0.3} \\
\rowcolor{groupgray}
\multicolumn{8}{c}{\textit{Foundation Embeddings}} \\
Gen-Rep     & 69.20{\footnotesize$\pm$0.1} & 68.10{\footnotesize$\pm$0.2} & 66.40{\footnotesize$\pm$0.1} & 67.30{\footnotesize$\pm$0.2} & 69.10{\footnotesize$\pm$0.2} & 67.90{\footnotesize$\pm$0.1} & 68.00{\footnotesize$\pm$0.1} \\
Audio-PC    & 72.30{\footnotesize$\pm$0.1} & 71.40{\footnotesize$\pm$0.2} & 69.50{\footnotesize$\pm$0.1} & 70.20{\footnotesize$\pm$0.1} & 73.10{\footnotesize$\pm$0.2} & 71.10{\footnotesize$\pm$0.2} & 71.30{\footnotesize$\pm$0.1} \\
\rowcolor{groupgray}
\multicolumn{8}{c}{\textit{Unified Models}} \\
CADA        & 62.50{\footnotesize$\pm$0.7} & 65.30{\footnotesize$\pm$0.8} & 67.80{\footnotesize$\pm$0.6} & 60.80{\footnotesize$\pm$0.7} & 62.10{\footnotesize$\pm$0.6} & 66.10{\footnotesize$\pm$0.5} & 64.10{\footnotesize$\pm$0.6} \\
UniFlow     & 64.80{\footnotesize$\pm$0.5} & 68.20{\footnotesize$\pm$0.4} & 66.90{\footnotesize$\pm$0.6} & 64.80{\footnotesize$\pm$0.5} & 67.50{\footnotesize$\pm$0.4} & 66.50{\footnotesize$\pm$0.6} & 66.40{\footnotesize$\pm$0.5} \\
Dinomaly    & 66.80{\footnotesize$\pm$0.4} & 64.20{\footnotesize$\pm$0.5} & 64.40{\footnotesize$\pm$0.3} & 67.80{\footnotesize$\pm$0.4} & 65.20{\footnotesize$\pm$0.4} & 68.40{\footnotesize$\pm$0.3} & 66.20{\footnotesize$\pm$0.4} \\
\rowcolor{groupgray}
\multicolumn{8}{c}{\textit{Disentanglement Baselines}} \\
GRHD        & 60.10{\footnotesize$\pm$0.3} & 59.20{\footnotesize$\pm$0.3} & 60.40{\footnotesize$\pm$0.2} & 58.30{\footnotesize$\pm$0.4} & 61.20{\footnotesize$\pm$0.3} & 60.50{\footnotesize$\pm$0.4} & 60.00{\footnotesize$\pm$0.3} \\
MCDD        & 67.40{\footnotesize$\pm$0.9} & 66.50{\footnotesize$\pm$0.7} & 65.20{\footnotesize$\pm$0.8} & 67.10{\footnotesize$\pm$0.8} & 68.30{\footnotesize$\pm$0.9} & 67.00{\footnotesize$\pm$0.7} & 66.90{\footnotesize$\pm$0.8} \\
AudDSR      & 69.20{\footnotesize$\pm$1.2} & 68.50{\footnotesize$\pm$1.0} & 66.30{\footnotesize$\pm$1.3} & 68.10{\footnotesize$\pm$1.1} & 70.40{\footnotesize$\pm$0.9} & 69.10{\footnotesize$\pm$1.2} & 68.60{\footnotesize$\pm$1.1} \\
\rowcolor{groupgray}
\multicolumn{8}{c}{\textit{Proposed Method}} \\
\rowcolor{harmonyblue}
HARMONY     & \textbf{90.20}{\footnotesize$\pm$0.1} & \textbf{71.80}{\footnotesize$\pm$0.2} & \textbf{68.40}{\footnotesize$\pm$0.2} & \textbf{75.00}{\footnotesize$\pm$0.2} & \textbf{74.20}{\footnotesize$\pm$0.2} & \textbf{69.60}{\footnotesize$\pm$0.3} & \textbf{74.90}{\footnotesize$\pm$0.2} \\
\bottomrule
\end{tabular}
\end{table*}
\FloatBarrier
\section{Direct Control Experiments}
\label{app:direct_control}

This section isolates the source of performance gains by comparing HARMONY against specific methodological permutations. \cref{tab:direct_control} demonstrates that applying standard approximate indexing (IVF/HNSW) directly to raw features yields no performance benefit due to manifold entanglement. Furthermore, projection into an unbounded Euclidean space ($\mathbb{R}^{128}$) rather than the unit hypersphere results in substantial degradation, confirming the structural advantage of the spherical constraint.

\begin{table}[ht]
\normalsize
\renewcommand{\arraystretch}{1.15}
\setlength{\tabcolsep}{5pt}
\caption{Direct control experiments validating the necessity of individual architectural choices. Metrics are in percentages (\%).}
\label{tab:direct_control}
\centering
\begin{tabular}{l ccccccc}
\toprule
\multirow{2}{*}{\textbf{Method Variant}} & \multicolumn{6}{c}{\textbf{Machine Types}} & \multirow{2}{*}{\textbf{Avg. AUC}} \\
\cmidrule(lr){2-7}
& ToyCar & ToyConv & Fan & Pump & Slider & Valve & \\
\midrule
Raw Feat + Exact kNN & 64.10 & 56.50 & 58.80 & 69.20 & 62.30 & 63.50 & 62.40 \\
PCA + KMeans + Ret.  & 66.50 & 58.20 & 61.00 & 71.50 & 64.80 & 66.20 & 64.70 \\
Raw Feat + IVF/HNSW  & 63.50 & 56.10 & 58.40 & 68.50 & 61.90 & 63.10 & 61.90 \\
Proj + Random Init   & 89.50 & 76.50 & 78.20 & 82.50 & 79.00 & 83.10 & 81.40 \\
Euclidean Proj ($\mathbb{R}^{128}$) & 86.20 & 72.10 & 74.50 & 80.10 & 76.80 & 79.50 & 78.20 \\
\midrule
Proj + Exact kNN (Flat) & 98.00 & 89.70 & 88.70 & 89.70 & 87.60 & 90.00 & 90.00 \\
\rowcolor{harmonyblue}
\textbf{HARMONY (IVF)} & \textbf{97.80} & \textbf{89.50} & \textbf{88.50} & \textbf{89.50} & \textbf{87.40} & \textbf{89.80} & \textbf{89.80} \\
\bottomrule
\end{tabular}
\end{table}
\subsection{Matched Prototype, Clustering, and Index Controls}
\cref{tab:matched_controls} compares routed detection AUC, nearest-neighbor recall, and CPU latency over five seeds. Prototype and clustering controls use local nearest-neighbor scoring. HARMONY rep.+HNSW applies an independent index to the learned representation. Its AUC is close to full HARMONY, whereas shared routing raises recall from 97.26\% to 98.49\%. Thus, these results support separating the representation gain from the routing effect. The hard-assignment control here is evaluated with hierarchical routing; the flat-search ablation in \cref{tab:ablation} is a distinct setting.
\begin{table}[htbp]
\centering
\small
\setlength{\tabcolsep}{4pt}
\caption{Matched controls on DCASE 2020. AUC and recall are percentages; latency is milliseconds.}
\label{tab:matched_controls}
\begin{tabular}{lccc}
\toprule
Method & Routed AUC & Recall & Latency \\
\midrule
PCA+KMeans & $64.71\pm0.29$ & $87.94\pm0.83$ & $6.81\pm0.28$ \\
MCDD+local kNN & $81.47\pm0.87$ & $94.18\pm0.62$ & $7.26\pm0.33$ \\
DeepCluster+local kNN & $84.62\pm0.72$ & $95.73\pm0.54$ & $6.32\pm0.27$ \\
vMF soft prototypes & $85.71\pm0.61$ & $96.21\pm0.49$ & $6.19\pm0.26$ \\
HARMONY hard & $87.34\pm0.49$ & $96.77\pm0.43$ & $5.57\pm0.23$ \\
HARMONY rep.+HNSW & $89.71\pm0.32$ & $97.26\pm0.34$ & $5.16\pm0.22$ \\
HARMONY full & $89.80\pm0.31$ & $98.49\pm0.27$ & $5.40\pm0.22$ \\
\bottomrule
\end{tabular}
\end{table}

\FloatBarrier
\section{Stress Tests under Extreme Conditions}
\label{app:stress_tests}

Our theoretical framework models informative signals within an orthogonal subspace. We conducted stress tests to observe degradation patterns when these assumptions are intentionally violated. The results indicate a gradual performance degradation, demonstrating robustness even when ideal geometric conditions are compromised. Specifically, \cref{fig:stress_tests} visualizes the performance degradation under non-linear manifold distortion and PCA-aligned anisotropic noise, showing that HARMONY maintains significant structural advantages over passive indexing even in severe scenarios.

\begin{figure}[ht]
    \centering
    \includegraphics[width=\linewidth]{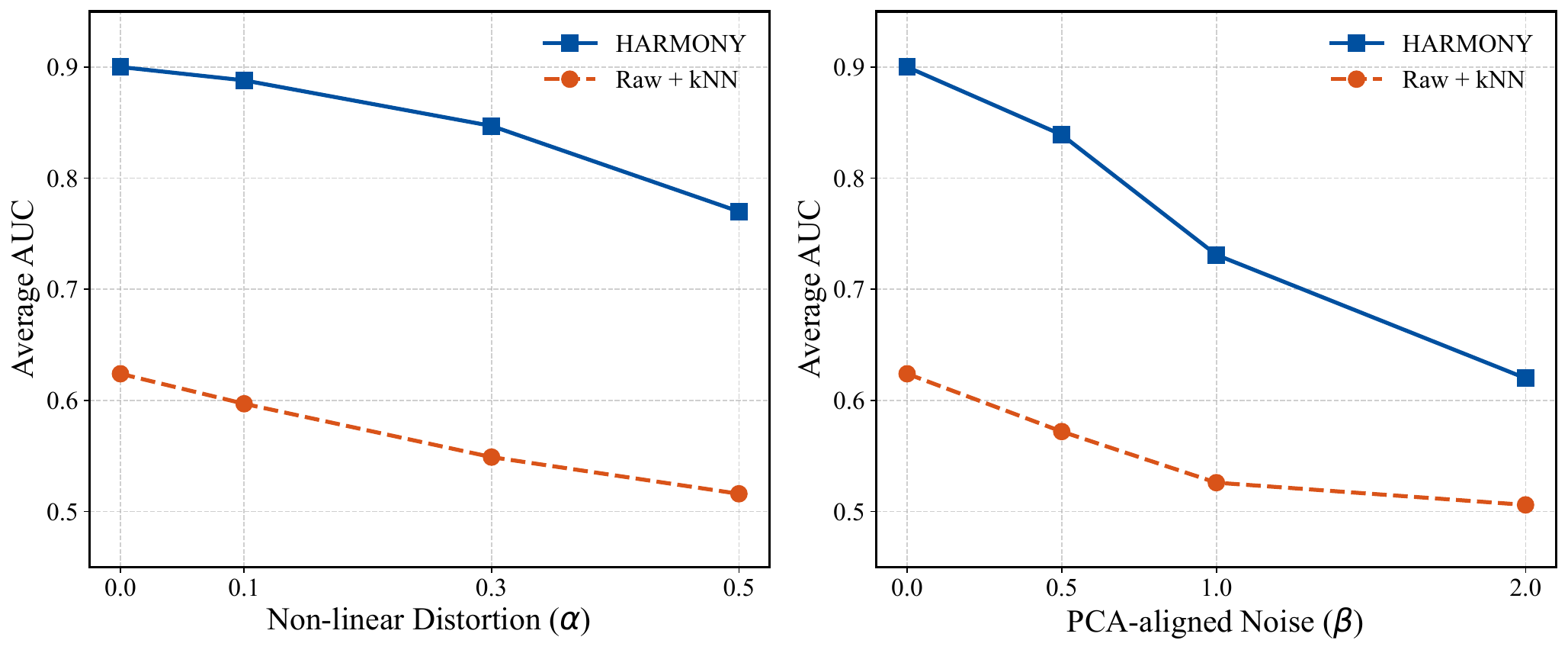}
    \caption{Performance degradation under severe manifold perturbations. (Left) Non-linear distortion eroding the geometric margin. (Right) PCA-aligned anisotropic noise violating orthogonal assumptions. HARMONY exhibits graceful decay compared to the baseline.}
    \label{fig:stress_tests}
\end{figure}
\subsection{Recorded Environmental Noise}
DCASE 2020 test clips are mixed with recorded environmental noise from DEMAND \citep{thiemann2013demand}. The evaluation covers fixed SNRs of $+6$, 0, and $-6$ dB, within-clip varying SNR, and environments excluded from model selection. All methods receive the same noisy audio in each of five seeds. \cref{tab:demand_noise} reports AUC; HARMONY exceeds Audio-PC by 13.31 points at $-6$ dB, 11.28 under dynamic SNR, and 11.85 in unseen environments.
\begin{table}[htbp]
\centering
\small
\setlength{\tabcolsep}{4pt}
\caption{AUC under recorded environmental noise (\%, mean $\pm$ standard deviation).}
\label{tab:demand_noise}
\begin{tabular}{lcccccc}
\toprule
Method & Original & $+6$ dB & 0 dB & $-6$ dB & Dynamic & Unseen \\
\midrule
Audio-PC & $83.90\pm0.20$ & $80.82\pm0.34$ & $75.47\pm0.51$ & $66.11\pm0.76$ & $71.84\pm0.64$ & $69.52\pm0.71$ \\
UniFlow & $75.40\pm0.60$ & $72.17\pm0.74$ & $66.32\pm0.91$ & $58.41\pm1.18$ & $63.18\pm1.02$ & $61.02\pm1.11$ \\
Dinomaly & $75.30\pm0.50$ & $72.64\pm0.63$ & $67.28\pm0.79$ & $60.17\pm1.04$ & $64.22\pm0.89$ & $62.34\pm0.97$ \\
HARMONY & $90.00\pm0.30$ & $88.73\pm0.42$ & $85.64\pm0.62$ & $79.42\pm0.89$ & $83.12\pm0.75$ & $81.37\pm0.82$ \\
\bottomrule
\end{tabular}
\end{table}

\FloatBarrier
\section{Hyperparameter Sensitivity and Unsupervised Selection}
\label{app:hyperparameters}

A practical anomaly detection system must maintain stability across broad configuration ranges. \cref{fig:sensitivity} evaluates the sensitivity of the model to the anchor count $K$, intrinsic dimensionality $d'$, and temperature $\tau$, demonstrating wide performance plateaus without requiring extensive numerical tuning. Furthermore, \cref{fig:silhouette_correlation} substantiates the strong global correlation between the unsupervised Silhouette Score and test AUC across all these variations and stress tests. This strong linear alignment provides a rigorous empirical basis for label-free model selection.

\begin{figure}[ht]
    \centering
    \includegraphics[width=\linewidth]{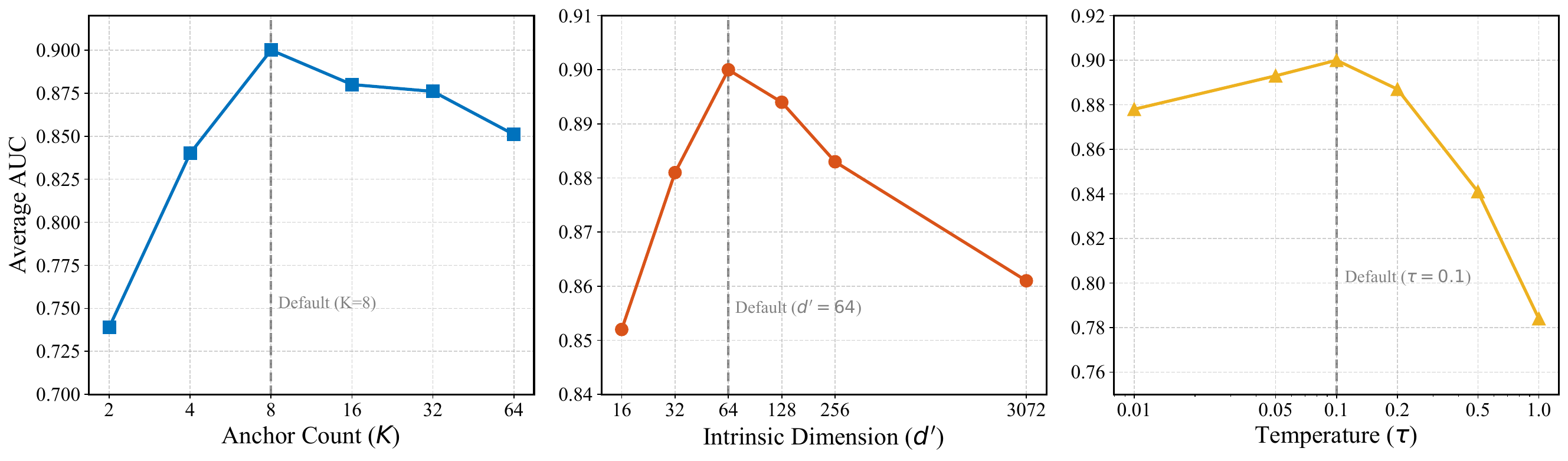}
    \caption{Hyperparameter sensitivity analysis. The framework demonstrates robust performance plateaus across broad ranges of anchor count $K$, intrinsic dimensionality $d'$, and softmax temperature $\tau$. The default configurations are situated within optimal stable regions.}
    \label{fig:sensitivity}
\end{figure}

\begin{figure}[ht]
    \centering
    \includegraphics[width=0.65\linewidth]{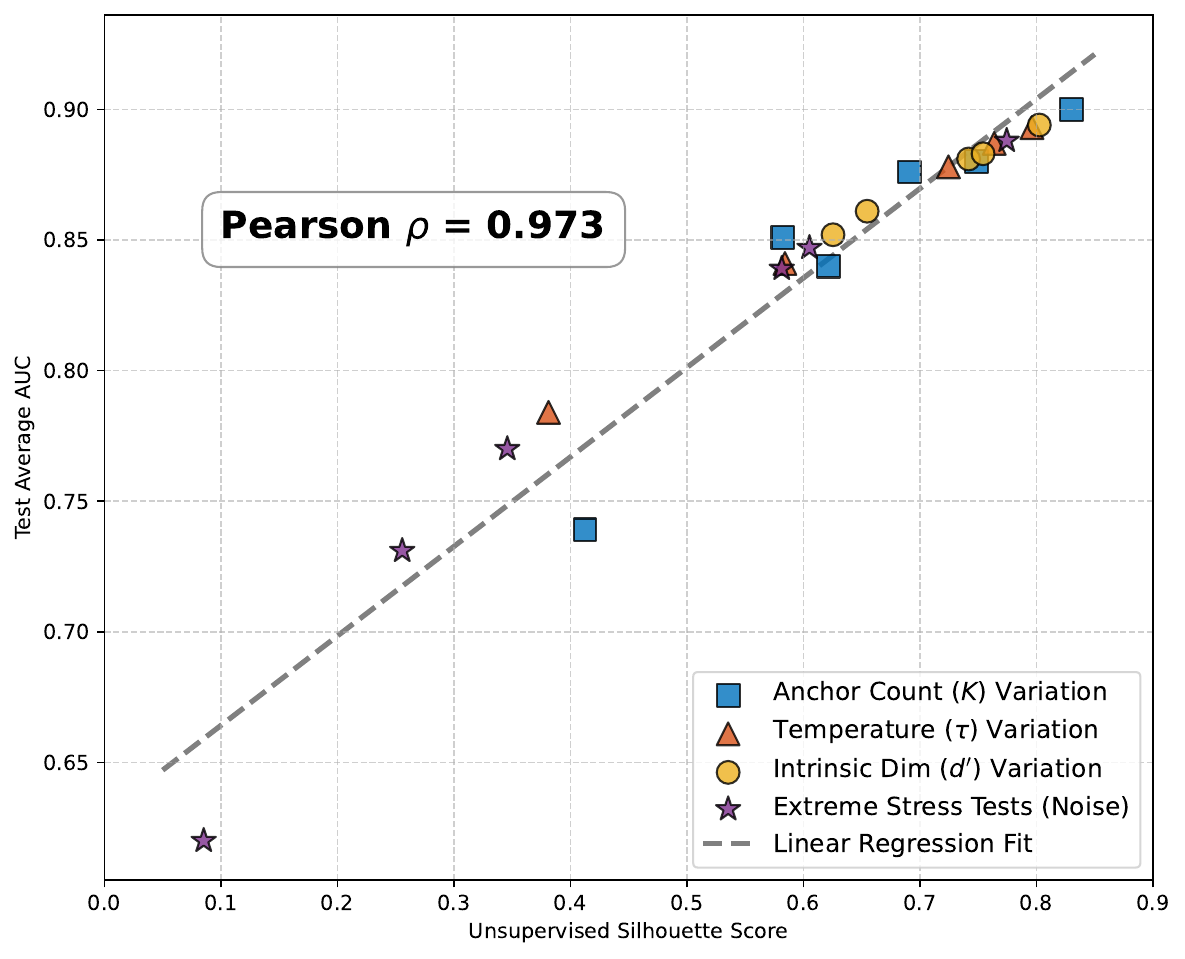}
    \caption{Global correlation between the unsupervised Silhouette Score and test Average AUC across all hyperparameter variations and extreme stress tests. The strong linear alignment (Pearson $\rho = 0.973$) validates the use of the Silhouette Score as a highly reliable surrogate metric for label-free model selection.}
    \label{fig:silhouette_correlation}
\end{figure}
\subsection{Selection Across Machine Mixtures}
For each candidate anchor count $K$, a partition is learned using pooled normal data and scored by its Silhouette. The highest-scoring candidate is selected without anomaly labels. Test-oracle AUC is used only for retrospective analysis. \cref{tab:mixture_selection} reports the mode and range of selected counts across five seeds; $\Delta_{oracle}$ is the AUC difference from the test oracle, and Spearman $\rho_s$ compares the rankings of candidates by Silhouette and AUC. The selected count adapts to mixture composition while remaining close to the oracle.
\begin{table}[htbp]
\centering
\small
\setlength{\tabcolsep}{4pt}
\caption{Normal-only selection for mixtures of $C$ machine types (AUC, \%).}
\label{tab:mixture_selection}
\begin{tabular}{lccccc}
\toprule
$C$ & Selected $K$ & Oracle $K$ & Selected / oracle AUC & $\Delta_{oracle}$ & $\rho_s$ \\
\midrule
2 & 4 / 2--6 & 4 / 4--8 & 92.84 / 93.11 & 0.27 & 0.91 \\
3 & 6 / 4--8 & 6 / 6--8 & 91.97 / 92.29 & 0.32 & 0.93 \\
4 & 8 / 6--12 & 8 / 8--12 & 91.31 / 91.58 & 0.27 & 0.94 \\
5 & 8 / 8--12 & 12 / 8--12 & 90.64 / 91.01 & 0.37 & 0.96 \\
6 & 8 / 8--12 & 8 / 8--16 & 90.07 / 90.26 & 0.19 & 0.97 \\
\bottomrule
\end{tabular}
\end{table}

The same criterion is applied to pooled MIMII machine IDs in \cref{tab:mimii_selection}. It selects a modal count of 12 rather than transferring the DCASE count of 8, reducing the test-oracle gap from 0.63 to 0.27 AUC points. These results correspond to the additional MIMII comparison in Appendix \ref{app:cross_dataset}.
\begin{table}[htbp]
\centering
\small
\setlength{\tabcolsep}{4pt}
\caption{MIMII selection and retrospective oracle comparison.}
\label{tab:mimii_selection}
\begin{tabular}{lccc}
\toprule
Setting & $K$ (mode / range) & AUC (\%) & $\Delta_{oracle}$ \\
\midrule
DCASE count transferred & 8 / 8--8 & $84.21\pm0.63$ & $0.63\pm0.28$ \\
Normal-only selection & 12 / 8--16 & $84.57\pm0.58$ & $0.27\pm0.19$ \\
Test oracle & 12 / 12--16 & $84.84\pm0.51$ & $0.00\pm0.00$ \\
\bottomrule
\end{tabular}
\end{table}

\subsection{Long-Tailed Mixtures and Temperature}
The imbalance ratio is $IR=n_{max}/n_{min}$ for normal-training machine-type counts; the test protocol is held fixed. For each of five seeds, $K$ and $\tau$ are selected by normal-data Silhouette. Tail AUC is the least frequent type's AUC, Head--Tail Gap is the most frequent type's AUC minus the tail AUC, and Entropy Ratio is anchor-occupancy entropy divided by $\log K$. As shown in \cref{tab:long_tail}, increasing imbalance reduces tail performance and occupancy entropy. At $IR\geq10$, the selected configuration uses more anchors and a smaller temperature.
\begin{table}[htbp]
\centering
\small
\setlength{\tabcolsep}{4pt}
\caption{Joint selection under long-tailed training mixtures. AUC and gaps use percentage-point units; entropy is dimensionless.}
\label{tab:long_tail}
\begin{tabular}{lcccccc}
\toprule
$IR$ & $K/\tau$ & Macro AUC & Tail AUC & Head--Tail & Entropy & $\Delta_{oracle}$ \\
\midrule
1 & 8 / 0.10 & $90.00\pm0.30$ & $89.78\pm0.42$ & $0.34\pm0.31$ & $0.98\pm0.01$ & $0.18\pm0.14$ \\
3 & 8 / 0.10 & $89.31\pm0.44$ & $88.16\pm0.68$ & $1.88\pm0.55$ & $0.96\pm0.02$ & $0.28\pm0.20$ \\
5 & 8 / 0.10 & $88.54\pm0.57$ & $86.63\pm0.87$ & $3.23\pm0.73$ & $0.93\pm0.02$ & $0.37\pm0.25$ \\
10 & 16 / 0.05 & $86.96\pm0.82$ & $83.41\pm1.24$ & $5.76\pm1.08$ & $0.89\pm0.03$ & $0.58\pm0.34$ \\
20 & 16 / 0.05 & $84.83\pm1.06$ & $79.36\pm1.72$ & $8.97\pm1.46$ & $0.83\pm0.04$ & $0.84\pm0.47$ \\
\bottomrule
\end{tabular}
\end{table}

To isolate temperature, \cref{tab:temperature_fixed} fixes $K=8$. A global temperature does not change maximum-similarity boundaries for fixed representations, but it changes the training probabilities and gradients. The smaller temperature is less effective on balanced data and more effective under the evaluated imbalance. These results motivate the density-aware extension discussed in Appendix \ref{app:limitations_impacts}.
\begin{table}[htbp]
\centering
\small
\setlength{\tabcolsep}{4pt}
\caption{Fixed-$K$ temperature comparison (AUC, \%; $K=8$).}
\label{tab:temperature_fixed}
\begin{tabular}{lccc}
\toprule
$IR$ & $\tau=0.05$ & $\tau=0.10$ & Difference \\
\midrule
1 & $89.34\pm0.35$ & $90.00\pm0.30$ & $-0.66$ \\
10 & $86.44\pm0.79$ & $85.91\pm0.91$ & $+0.53$ \\
20 & $84.11\pm1.04$ & $82.86\pm1.27$ & $+1.25$ \\
\bottomrule
\end{tabular}
\end{table}

\subsection{Rare-Type Absorption}
The two largest machine types are designated frequent and the two smallest rare. Absorption is the fraction of rare clips assigned to anchors dominated by frequent types. \cref{tab:rare_absorption} shows that rare-type separation deteriorates with imbalance, although HARMONY retains more rare-type structure than the compared alternatives at $IR=20$. At that imbalance, changing $K$ from 8 to the selected 16 reduces absorption from 0.39 to 0.29 and raises rare-type AUC from 75.91\% to 79.36\%.
\begin{table}[htbp]
\centering
\small
\setlength{\tabcolsep}{4pt}
\caption{Rare-type behavior (five-seed mean $\pm$ standard deviation).}
\label{tab:rare_absorption}
\begin{tabular}{lcccc}
\toprule
Method / $IR$ & $K$ & Rare-type anchors & Absorption & Rare-type AUC (\%) \\
\midrule
HARMONY / 1 & 8 & $2.67\pm0.31$ & $0.05\pm0.02$ & $89.78\pm0.42$ \\
HARMONY / 5 & 8 & $2.14\pm0.26$ & $0.14\pm0.04$ & $86.63\pm0.87$ \\
HARMONY / 10 & 16 & $1.87\pm0.24$ & $0.18\pm0.04$ & $83.41\pm1.24$ \\
HARMONY / 20 & 16 & $1.19\pm0.22$ & $0.29\pm0.05$ & $79.36\pm1.72$ \\
MCDD / 20 & 16 & $0.87\pm0.21$ & $0.47\pm0.07$ & $66.21\pm1.48$ \\
PCA-KMeans / 20 & 8 & $0.23\pm0.18$ & $0.64\pm0.08$ & $53.92\pm1.57$ \\
\bottomrule
\end{tabular}
\end{table}

\FloatBarrier
\section{Implementation and Baseline Adaptation Details}
\label{app:implementation}

\subsection{Network Architecture}
We use \textbf{BYOL-A} \citep{niizumi2022byol} as the frozen audio backbone, pre-trained on AudioSet \citep{gemmeke2017audio}.
\begin{itemize}
    \item \textbf{Input:} Audio samples are converted to Log-Mel spectrograms with dimensions $64 \times 96$ ($F \times T$).
    \item \textbf{Feature Aggregation:} Following PatchCore \citep{roth2022towards}, we extract and concatenate the globally average-pooled outputs from intermediate blocks (Block 2 and Block 3). This results in a feature vector $z \in \mathbb{R}^{D}$ where $D=3072$.
    \item \textbf{Freezing:} The backbone weights are frozen in all experiments to evaluate the efficacy of the proposed manifold learning on fixed representations.
\end{itemize}

The projection head $g_\theta$ is a Multi-Layer Perceptron (MLP) with a bottleneck structure mapping the $3072$-dim input to a $128$-dim intrinsic space:
\begin{equation}
    \mathbb{R}^{3072} \xrightarrow{1024} \mathbb{R}^{1024} \xrightarrow{512} \mathbb{R}^{512} \xrightarrow{128} \mathbb{R}^{128}.
\end{equation}
The final output is $\ell_2$-normalized to the unit hypersphere $\mathbb{S}^{127}$.

\subsection{Principled Baseline Adaptation Setup}
To ensure fair comparison under the universal protocol without weakening the baselines, all methods utilized the identical frozen BYOL-A features and were tuned within equivalent hyperparameter budgets. The precise adaptations to handle high-dimensional acoustic manifolds are detailed below:
\begin{itemize}
    \item \textbf{PatchCore:} Retained Coreset-based k-NN. PCA (3072$\rightarrow$256) was utilized solely for memory-efficient coreset sampling, while the final $L_2$ distance was strictly calculated in the original 3072D space.
    \item \textbf{EfficientAD:} Removed image-teacher modality bias. Utilized frozen BYOL-A as the Teacher. The Student network was expanded to 2048D hidden layers to accommodate high-dimensional features.
    \item \textbf{SimpleNet:} Discriminator adapted to Linear(3072$\rightarrow$1024) $\rightarrow$ BN $\rightarrow$ LeakyReLU. Noise $\sigma$ searched within \{0.01, 0.05, 0.1\}.
    \item \textbf{CADA:} Projector adapted to 3072$\rightarrow$1024$\rightarrow$128 $\rightarrow$ L2 Norm, optimizing a learnable single 128D center.
    \item \textbf{UniFlow:} Normalizing flows configured with deterministic PCA (3072$\rightarrow$256) for stability, utilizing 8-layer Coupling Layers (256$\rightarrow$512$\rightarrow$256).
    \item \textbf{Dinomaly:} Adapter (3072$\rightarrow$512) followed by a 2-layer Encoder and 2-layer Decoder (4 heads), maintaining the strict bottleneck scaling of the original framework.
    \item \textbf{MCDD:} Projection optimized through 3072$\rightarrow$1024$\rightarrow$512 with 16 learnable prototypes via multi-class cross-entropy.
\end{itemize}

\subsection{Training Configuration}
We optimize using Adam ($\beta_1=0.9, \beta_2=0.999$, weight decay $10^{-4}$) governed by a Cosine Annealing scheduler ($10^{-3}$ to $10^{-5}$). \cref{tab:hyperparams} summarizes the default settings.

\begin{table}[ht] 
\normalsize
\renewcommand{\arraystretch}{1.15}
\setlength{\tabcolsep}{10pt}
\caption{Default hyperparameter settings.}
\label{tab:hyperparams}
\centering
\begin{tabular}{llcl}
\toprule
Category & Parameter & Symbol & Value \\
\midrule
\textbf{Data} & Sample Rate & $sr$ & 16,000 Hz \\
 & FFT / Hop & - & 1024 / 160 \\
 & Mel Bins & $N_{mels}$ & 64 \\
\midrule
\textbf{Model} & Input Dim & $D$ & 3072 \\
 & Latent Dim & $d$ & 128 \\
\midrule
\textbf{Train} & Epochs & $E$ & 20 (DCASE) \\
 & Batch Size & $B_{sz}$ & 128 \\
 & Learning Rate & $\eta$ & $1 \times 10^{-3}$ \\
\midrule
\textbf{Method} & Anchors & $K$ & 8 \\
 & Temperature & $\tau$ & 0.1 \\
 & Buckets & $B$ & 64 \\
\bottomrule
\end{tabular}
\end{table}

\subsection{Environment}
Inference latency is measured strictly on an independent server to avoid resource contention.
\begin{itemize}
    \item \textbf{Hardware:} Training on NVIDIA Tesla V100. Edge efficiency tests on an Intel Xeon Platinum 8163 CPU.
    \item \textbf{Latency Protocol:} CPU inference is restricted to a single thread. The reported latency encapsulates the entire forward pass, including feature extraction via the frozen BYOL-A backbone, representing the average processing time per single audio sample.
    \item \textbf{Software:} PyTorch 1.13.1 and Faiss 1.7.4.
\end{itemize}

\FloatBarrier
\section{Datasets and Evaluation Protocol}
\label{app:datasets}

\subsection{Dataset Details}
We evaluate on \textbf{DCASE 2020 Task 2} and the \textbf{MIMII Dataset (-6dB subset)}. \cref{tab:dcase_stats} and \cref{tab:mimii_stats} detail the exact sample distributions.

\begin{table}[ht]
\normalsize
\renewcommand{\arraystretch}{1.15}
\setlength{\tabcolsep}{8pt}
\caption{Detailed statistics of the DCASE 2020 Universal Dataset.}
\label{tab:dcase_stats}
\centering
\begin{tabular}{llccc}
\toprule
Machine Type & Machine IDs & Train (Normal) & Test (Normal) & Test (Anomaly) \\
\midrule
ToyCar      & 01, 02, 03, 04 & 4,000 & 1,400 & 1,059 \\
ToyConveyor & 01, 02, 03     & 3,000 & 2,399 & 1,110 \\
Fan         & 00, 02, 04, 06 & 3,675 & 400   & 1,475 \\
Pump        & 00, 02, 04, 06 & 3,349 & 400   & 456   \\
Slider      & 00, 02, 04, 06 & 2,804 & 400   & 890   \\
Valve       & 00, 02, 04, 06 & 3,291 & 400   & 479   \\
\midrule
\textbf{Total} & \textbf{23 IDs} & \textbf{20,119} & \textbf{5,399} & \textbf{5,469} \\
\bottomrule
\end{tabular}
\end{table}

\begin{table}[ht]
\normalsize
\renewcommand{\arraystretch}{1.15}
\setlength{\tabcolsep}{8pt}
\caption{Detailed statistics of the MIMII -6dB Dataset.}
\label{tab:mimii_stats}
\centering
\begin{tabular}{llccc}
\toprule
Machine Type & Machine IDs & Train (Normal) & Test (Normal) & Test (Anomaly) \\
\midrule
Fan    & 00, 02, 04, 06 & 3,261 & 814 & 1,475 \\
Pump   & 00, 02, 04, 06 & 3,000 & 749 & 456   \\
Slider & 00, 02, 04, 06 & 2,566 & 638 & 890   \\
Valve  & 00, 02, 04, 06 & 2,954 & 737 & 479   \\
\midrule
\textbf{Total} & \textbf{16 IDs} & \textbf{11,781} & \textbf{2,938} & \textbf{3,300} \\
\bottomrule
\end{tabular}
\end{table}

\subsection{Universal Unsupervised Protocol Definition}
To standardize the evaluation:
\begin{itemize}
    \item \textbf{Label Blindness:} The training algorithm input is strictly $\{x_i\}$, excluding any machine type or ID labels.
    \item \textbf{Single Model:} All data is processed by a single set of network parameters $\theta$ and a single index structure $\mathcal{T}$.
    \item \textbf{Standardized Eval:} During testing, anomaly scores are computed for each sample, and the AUC is aggregated by grouping samples based on ground-truth machine labels, reporting the Macro-Average AUC.
\end{itemize}
\FloatBarrier
\section{Detailed Experimental Results}
\label{app:detailed_results}

This section provides itemized performance metrics for each machine ID across both benchmarks.

\subsection{DCASE 2020 Per-ID Performance}
\cref{tab:dcase_per_id} reports the specific AUC metrics for all 23 machine IDs in the DCASE 2020 dataset.

\begin{table}[ht]
\normalsize
\renewcommand{\arraystretch}{1.15}
\setlength{\tabcolsep}{5pt}
\caption{Per-ID performance on DCASE 2020. Metrics are reported in percentages (\%).}
\label{tab:dcase_per_id}
\centering
\begin{tabular}{lcccc|lcccc}
\toprule
\textbf{Type} & \textbf{ID} & \textbf{AUC} & \textbf{pAUC} & \textbf{Time} & 
\textbf{Type} & \textbf{ID} & \textbf{AUC} & \textbf{pAUC} & \textbf{Time} \\
\midrule
ToyCar      & 01 & 97.78 & 88.89 & 5.39 ms & Pump        & 00 & 89.08 & 74.53 & 5.34 ms \\
            & 02 & 97.97 & 90.51 & 5.12 ms &             & 02 & 85.50 & 65.39 & 5.11 ms \\
            & 03 & 97.77 & 88.80 & 5.46 ms &             & 04 & 89.40 & 80.68 & 5.64 ms \\
            & 04 & 98.42 & 92.42 & 5.33 ms &             & 06 & 90.77 & 79.51 & 5.57 ms \\
\cmidrule(r){1-5} \cmidrule(l){6-10}
ToyConveyor & 01 & 92.15 & 75.47 & 5.06 ms & Slider      & 00 & 93.52 & 80.53 & 5.65 ms \\
            & 02 & 88.94 & 71.41 & 5.16 ms &             & 02 & 89.51 & 72.26 & 5.62 ms \\
            & 03 & 88.04 & 68.55 & 5.14 ms &             & 04 & 90.86 & 75.96 & 5.60 ms \\
\cmidrule(r){1-5}
Fan         & 00 & 89.88 & 75.07 & 5.31 ms &             & 06 & 84.87 & 67.95 & 5.13 ms \\
\cmidrule(l){6-10}
            & 02 & 88.16 & 70.63 & 5.27 ms & Valve       & 00 & 89.68 & 74.17 & 5.38 ms \\
            & 04 & 84.25 & 62.60 & 5.87 ms &             & 02 & 88.18 & 67.72 & 5.26 ms \\
            & 06 & 82.52 & 65.45 & 5.36 ms &             & 04 & 88.23 & 70.18 & 5.70 ms \\
            &    &       &       &         &             & 06 & 84.29 & 66.32 & 5.95 ms \\
\midrule
\rowcolor{harmonyblue}
\multicolumn{5}{c|}{} & \textbf{Mean (All)} & - & \textbf{90.00} & \textbf{75.00} & \textbf{5.41 ms} \\
\bottomrule
\end{tabular}
\end{table}

\subsection{MIMII -6dB Per-ID Performance}
\cref{tab:mimii_per_id} reports the specific AUC metrics for the MIMII -6dB dataset.

\begin{table}[ht]
\normalsize
\renewcommand{\arraystretch}{1.15}
\setlength{\tabcolsep}{5pt}
\caption{Per-ID performance on MIMII -6dB. Metrics are reported in percentages (\%).}
\label{tab:mimii_per_id}
\centering
\begin{tabular}{lcccc|lcccc}
\toprule
\textbf{Type} & \textbf{ID} & \textbf{AUC} & \textbf{pAUC} & \textbf{Time} & 
\textbf{Type} & \textbf{ID} & \textbf{AUC} & \textbf{pAUC} & \textbf{Time} \\
\midrule
Fan         & 00 & 84.54 & 65.68 & 5.63 ms & Slider      & 00 & 90.72 & 78.32 & 5.24 ms \\
            & 02 & 83.96 & 64.73 & 5.92 ms &             & 02 & 86.24 & 67.76 & 5.39 ms \\
            & 04 & 81.16 & 63.10 & 5.93 ms &             & 04 & 84.23 & 69.59 & 5.38 ms \\
            & 06 & 70.21 & 61.68 & 5.21 ms &             & 06 & 82.13 & 66.08 & 5.79 ms \\
\cmidrule(r){1-5} \cmidrule(l){6-10}
Pump        & 00 & 90.23 & 72.59 & 5.48 ms & Valve       & 00 & 80.72 & 62.17 & 5.18 ms \\
            & 02 & 89.20 & 73.27 & 5.61 ms &             & 02 & 83.95 & 67.34 & 5.10 ms \\
            & 04 & 86.33 & 75.38 & 5.56 ms &             & 04 & 87.80 & 70.35 & 5.62 ms \\
            & 06 & 87.77 & 71.29 & 5.06 ms &             & 06 & 85.39 & 70.95 & 5.17 ms \\
\midrule
\rowcolor{harmonyblue}
\multicolumn{5}{c|}{} & \textbf{Mean (All)} & - & \textbf{84.70} & \textbf{69.10} & \textbf{5.45 ms} \\
\bottomrule
\end{tabular}
\end{table}
\FloatBarrier
\section{Limitations and Broader Impacts}
\label{app:limitations_impacts}

\paragraph{Limitations.} 
While HARMONY achieves state-of-the-art performance in metadata-free machinery acoustic anomaly detection, it naturally inherits certain limitations from its geometric assumptions. First, the method relies on a pre-trained foundation feature space (e.g., BYOL-A) to provide meaningful initial local distances. If the domain shift of the target acoustic environment is overwhelmingly extreme (e.g., entirely unrepresented in the pre-training data), the initial topological skeleton extraction may become suboptimal. Second, the soft assignment mechanism currently utilizes a uniform temperature scaling factor $\tau$. For highly imbalanced mixed-source scenarios where some machine acoustic modes are drastically denser than others, incorporating an adaptive, density-aware temperature parameter could further improve the Voronoi partition margins. We plan to explore density-aware projection functions in future work.

\paragraph{Broader Societal Impacts.}
Our research primarily focuses on improving unsupervised acoustic anomaly detection (ASD) for industrial environments, aiming to identify machine failures reliably without relying on strict metadata. 
\textit{Positive impacts:} Implementing robust ASD systems like HARMONY can significantly enhance industrial safety, reduce catastrophic equipment failures, and lower environmental risks associated with industrial accidents. Furthermore, by alleviating the need for precise metadata labeling, it lowers the barrier to deploying predictive maintenance in resource-constrained environments.
\textit{Potential negative impacts:} As with many automated surveillance and monitoring technologies, there is a theoretical risk of misuse if the underlying acoustic feature extractors capture unintended sensitive auditory information (e.g., human speech or activities in the factory background). However, since our anomaly scoring relies exclusively on aggregated, anchor-aligned geometric density on a hypersphere rather than reconstructing raw audio, the risk of directly exposing raw sensitive content is intrinsically mitigated. We recommend that deployments of such systems adhere strictly to workplace privacy regulations and ensure that raw audio streams are not permanently stored.
\FloatBarrier
\section{Cross-Dataset Baseline Comparisons}
\label{app:cross_dataset}
All results in this section are mean $\pm$ standard deviation over five random seeds under a metadata-free machinery protocol. Normal training samples are pooled; machine labels are used for evaluation grouping. These additional five-seed evaluations are reported separately from the original per-ID results in Appendix \ref{app:detailed_results}; their aggregate values should not be combined. The normal-only MIMII selection is reported in \cref{tab:mimii_selection}.

\subsection{MIMII at $-6$ dB}
We compare all eleven baselines on a common MIMII split. Macro AUC and pAUC average groupwise scores; Worst-ID AUC is the minimum group AUC. The pAUC uses maximum false-positive rate 0.1. \cref{tab:mimii_baselines} shows improvements over Audio-PC of 4.84, 1.90, and 4.12 points, respectively.
\begin{table}[htbp]
\centering
\small
\setlength{\tabcolsep}{4pt}
\caption{MIMII $-6$ dB baseline comparison (\%).}
\label{tab:mimii_baselines}
\begin{tabular}{lccc}
\toprule
Method & Macro AUC & Macro pAUC & Worst-ID AUC \\
\midrule
PatchCore & $58.74\pm0.41$ & $53.92\pm0.36$ & $49.18\pm0.83$ \\
SimpleNet & $57.93\pm0.87$ & $52.84\pm0.72$ & $47.62\pm1.34$ \\
EfficientAD & $62.18\pm0.62$ & $56.73\pm0.51$ & $51.37\pm1.06$ \\
Gen-Rep & $77.86\pm0.48$ & $65.82\pm0.41$ & $63.94\pm0.92$ \\
Audio-PC & $79.73\pm0.44$ & $67.11\pm0.38$ & $65.72\pm0.87$ \\
CADA & $68.41\pm0.92$ & $59.28\pm0.76$ & $54.13\pm1.42$ \\
UniFlow & $71.24\pm0.79$ & $61.47\pm0.67$ & $56.84\pm1.28$ \\
Dinomaly & $72.06\pm0.71$ & $62.18\pm0.59$ & $57.61\pm1.17$ \\
GRHD & $61.37\pm0.66$ & $56.08\pm0.54$ & $50.72\pm1.11$ \\
MCDD & $75.48\pm1.04$ & $64.19\pm0.86$ & $60.81\pm1.63$ \\
AudDSR & $76.92\pm1.21$ & $65.04\pm1.02$ & $62.27\pm1.81$ \\
HARMONY & $84.57\pm0.58$ & $69.01\pm0.49$ & $69.84\pm1.27$ \\
\bottomrule
\end{tabular}
\end{table}

\subsection{DCASE 2022 Domain Shifts}
DCASE 2022 Task 2 \citep{dcase2022task2} introduces changes between source and target operating domains. Machine-type and section identifiers are withheld from model training and routing. Overall and target-domain AUC and the minimum section AUC are reported in \cref{tab:dcase2022_baselines}. HARMONY improves over Audio-PC by 4.52, 4.27, and 4.15 points, respectively.
\begin{table}[htbp]
\centering
\small
\setlength{\tabcolsep}{4pt}
\caption{DCASE 2022 comparison (AUC, \%).}
\label{tab:dcase2022_baselines}
\begin{tabular}{lccc}
\toprule
Method & Overall & Target & Worst section \\
\midrule
Audio-PC & $77.84\pm0.46$ & $73.72\pm0.59$ & $64.27\pm0.99$ \\
UniFlow & $72.69\pm0.81$ & $68.43\pm0.98$ & $58.92\pm1.48$ \\
Dinomaly & $73.21\pm0.73$ & $68.97\pm0.91$ & $59.41\pm1.39$ \\
HARMONY & $82.36\pm0.52$ & $77.99\pm0.68$ & $68.42\pm1.12$ \\
\bottomrule
\end{tabular}
\end{table}

\FloatBarrier
\section{Robustness to Training Contamination}
\label{app:contamination}
The standard setting in Section \ref{sec:preliminaries} assumes normal-only training data. To evaluate deviations from this assumption, anomalous clips are injected at contamination rates of 0.5\%, 1\%, 2\%, and 5\%. Every candidate contaminant is removed from the test set in all conditions, including the uncontaminated condition, preventing overlap between injected training anomalies and test clips. All methods are evaluated over five seeds. The uncontaminated baseline here therefore differs from the original DCASE 2020 experiment.

\cref{tab:contamination} shows that HARMONY drops by 0.61, 1.36, and 2.83 AUC points at 0.5\%, 1\%, and 2\% contamination, respectively. At 5\%, its AUC is 83.31\%, compared with 69.41\% for Audio-PC. Contamination still reduces performance, so the clean-training assumption remains relevant.
\begin{table}[htbp]
\centering
\small
\setlength{\tabcolsep}{4pt}
\caption{Training contamination on DCASE 2020. Columns are contamination rates; entries are AUC (\%, mean $\pm$ standard deviation).}
\label{tab:contamination}
\begin{tabular}{lccccc}
\toprule
Method & 0\% & 0.5\% & 1\% & 2\% & 5\% \\
\midrule
Raw+1NN & $62.18\pm0.23$ & $59.84\pm0.38$ & $57.63\pm0.51$ & $54.21\pm0.74$ & $50.47\pm1.06$ \\
PatchCore & $61.91\pm0.25$ & $60.42\pm0.36$ & $58.73\pm0.48$ & $55.97\pm0.69$ & $52.86\pm0.94$ \\
Audio-PC & $83.74\pm0.27$ & $82.16\pm0.39$ & $80.34\pm0.53$ & $76.88\pm0.77$ & $69.41\pm1.12$ \\
MCDD & $78.09\pm1.03$ & $77.18\pm1.08$ & $75.93\pm1.17$ & $73.47\pm1.31$ & $68.26\pm1.57$ \\
HARMONY & $89.65\pm0.35$ & $89.04\pm0.43$ & $88.29\pm0.57$ & $86.82\pm0.76$ & $83.31\pm1.13$ \\
\bottomrule
\end{tabular}
\end{table}

\end{document}